\documentclass[sigconf,screen]{acmart}
\copyrightyear{2026}
\acmYear{2026}
\setcopyright{cc}
\setcctype{by}
\acmConference[SPAA '26]{38th ACM Symposium on Parallelism in Algorithms and Architectures}{July 06--10, 2026}{London, United Kingdom}
\acmBooktitle{38th ACM Symposium on Parallelism in Algorithms and Architectures (SPAA '26), July 06--10, 2026, London, United Kingdom}
\acmDOI{10.1145/3816782.3819216}
\acmISBN{979-8-4007-2761-0/2026/07}

\usepackage{amsmath}
\usepackage{amsfonts}
\usepackage{thm-restate}
\usepackage{enumitem}
\usepackage{booktabs}
\usepackage{graphicx}
\usepackage{tabularx}
\usepackage{multirow}
\usepackage{xspace}
\usepackage[framemethod=tikz]{mdframed}

\usepackage[dvipsnames,svgnames]{xcolor}
\definecolor{mygreen}{RGB}{200, 255, 200}
\usepackage{hf-tikz}

\usepackage[linesnumbered, ruled, vlined]{algorithm2e}
\DontPrintSemicolon
\SetArgSty{textnormal}
\SetKwInOut{Input}{Input}
\SetKwInOut{Output}{Output}
\SetKwProg{function}{Function}{:}{}
\SetKwProg{parfor}{parallel for}{\:do}{}

\newtheorem{question}{Question}

\newtheorem{invariant}{Invariant}

\newcommand{\myparagraph}[1]{\vspace{.03in}\noindent {\bf #1.}}
\newcommand{\defn}[1]{{\textit{\textbf{\boldmath #1}}}}
\newcommand{\fname}[1]{\textsc{#1}}
\newcommand{\vname}[1]{\textit{#1}}
\newcommand{\oname}[1]{\mathsf{#1}}
\newcommand{\whp}{whp\xspace}

\newcommand{\myjoin}{\fname{Join}\xspace}
\newcommand{\mysplit}{\fname{Split}\xspace}
\newcommand{\getrep}{\fname{GetRepresentative}\xspace}
\newcommand{\getpred}{\fname{GetPredecessor}\xspace}
\newcommand{\getsucc}{\fname{GetSuccessor}\xspace}
\newcommand{\gethead}{\fname{GetHead}\xspace}
\newcommand{\gettail}{\fname{GetTail}\xspace}

\newcommand{\batchjoin}{\fname{BatchJoin}\xspace}
\newcommand{\batchsplit}{\fname{BatchSplit}\xspace}
\newcommand{\batchgetrep}{\fname{BatchGetRepresentative}\xspace}
\newcommand{\batchgetpred}{\fname{BatchGetPredecessor}\xspace}
\newcommand{\batchgetsucc}{\fname{BatchGetSuccessor}\xspace}
\newcommand{\batchgethead}{\fname{BatchGetHead}\xspace}
\newcommand{\batchgettail}{\fname{BatchGetTail}\xspace}

\newcommand{\find}{\fname{Find}\xspace}
\newcommand{\batchfind}{\fname{BatchFind}\xspace}

\newcommand{\batchupdate}{\fname{BatchUpdateWeight}\xspace}

\newcommand{\link}{\fname{Link}\xspace}
\newcommand{\cut}{\fname{Cut}\xspace}

\newcommand{\batchlink}{\fname{BatchLink}\xspace}
\newcommand{\batchcut}{\fname{BatchCut}\xspace}

\newcommand{\parent}{\vname{parent}\xspace}
\newcommand{\child}{\vname{child}\xspace}
\newcommand{\priority}{\vname{prio}\xspace}
\newcommand{\parentdirection}{\vname{dir\_from\_parent}\xspace}

\newcommand{\splice}{\fname{Splice}\xspace}
\newcommand{\expose}{\fname{Expose}\xspace}
\newcommand{\evert}{\fname{Evert}\xspace}
\newcommand{\reverse}{\fname{Reverse}\xspace}

\newcommand{\revision}[1]{#1}

\newif\ifsubmission
\newcommand{\appref}[1]{%
  \ifsubmission
    the full paper%
  \else
    #1%
  \fi
}

\begin{document}

\title{Fast and Theoretically-Efficient Batch-Parallel Link-Cut Trees, Euler Tour Trees, and Treaps}

\author{Quinten De Man}
\affiliation{
    \institution{University of Maryland}
    \city{College Park}
    \state{MD}
    \country{USA}
}
\email{deman@umd.edu}

\author{Laxman Dhulipala}
\affiliation{
    \institution{University of Maryland}
   \city{College Park}
   \state{MD}
   \country{USA}
}
\email{laxman@umd.edu}

\begin{abstract}

Parallel batch-dynamic trees are a fundamental building block in recent theoretical and practical advances in dynamic graph algorithms.
However, all existing parallel batch-dynamic tree data structures, including Euler tour trees, UFO trees, topology trees, and rake-compress trees, are all significantly outperformed in the sequential setting by link-cut trees, which have been the sequential state-of-the-art for over 40 years.
Despite their excellent performance in the sequential setting, designing efficient batch-parallel link-cut trees has remained a major open problem.

In this paper, we close this gap by introducing MOJOS, a unified framework for theoretically- and practically-efficient parallel batch-dynamic trees.
We exploit the fact that both Euler tour trees and link-cut trees rely on a common dynamic sequence abstraction that supports splitting and joining.
We introduce a new batch-dynamic sequence built using treaps that achieves optimal work and depth, and outperforms existing parallel skip list and treap implementations for batch updates, queries, and memory usage.

With MOJOS, we develop a new batch-parallel Euler tour tree algorithm that outperforms prior batch-dynamic tree implementations supporting subtree queries.
Unlike prior batch-parallel Euler tour trees which rely on skip list's ability to represent cyclic sequences, MOJOS allows any batch-dynamic sequence data structure to be used as a drop-in replacement.
Finally, we develop the first theoretically-efficient batch-parallel link-cut tree, which is also the first batch-dynamic data structure supporting path queries to achieve $O(\log n)$ depth for batch updates in the binary-forking model.
Our link-cut tree implementation outperforms all known parallel batch-dynamic tree data structures supporting path queries.

\end{abstract}



\maketitle

\section{Introduction}

Dynamic tree data structures maintain a forest of trees subject to edge insertions (links) and deletions (cuts), while supporting queries such as connectivity queries, path queries, and subtree queries.
Since their introduction in the seminal work of Sleator and Tarjan~\cite{sleator1983data}, dynamic trees have become an essential primitive in algorithms for network flows~\cite{chen2025maximum}, dynamic graph connectivity~\cite{holm2001poly}, and many other problems~\cite{holm2025fully}.
In recent years, \emph{parallel batch-dynamic} variants of dynamic trees---which process batches of links, cuts, and queries simultaneously---have emerged as a key ingredient in high-performance parallel dynamic graph algorithms~\cite{tseng2019batch, acar2019parallel, deman2026ufo, acar2020changeprop, anderson2024deterministic, ikram2025parallel}.
These batch-dynamic trees have enabled the construction of more sophisticated parallel batch-dynamic graph algorithms, including for connectivity~\cite{acar2019parallel, de2025towards}, minimum spanning forest~\cite{AndersonBT20}, graph clustering~\cite{deman2025fully, tseng2022parallel}, and other graph problems~\cite{ghaffari2023nearly, dhulipala2019parallel}.

Existing parallel batch-dynamic tree data structures fall into two broad families.
The first, based on Euler tour trees (ETTs)~\cite{henzinger1995randomized, tseng2019batch}, represent each tree in the forest as an Euler tour stored in a dynamic sequence, and only support subtree queries.
The second family, based on tree contraction~\cite{miller1985parallel}, includes topology trees~\cite{frederickson1985data,frederickson1997ambivalent,frederickson1997data,deman2026ufo}, rake-compress trees~\cite{acar2004dynamizing,acar2005experimental,acar2020changeprop,anderson2024deterministic,ikram2025parallel}, and, most recently, UFO trees~\cite{deman2026ufo}, and supports a rich set of queries, including path queries.
Since many applications require only one type of query (e.g., dynamic connectivity only requires augmented subtree queries), designing specialized data structures for each query type is important.

\begin{figure*}
    \centering
    \includegraphics[width=0.88\linewidth]{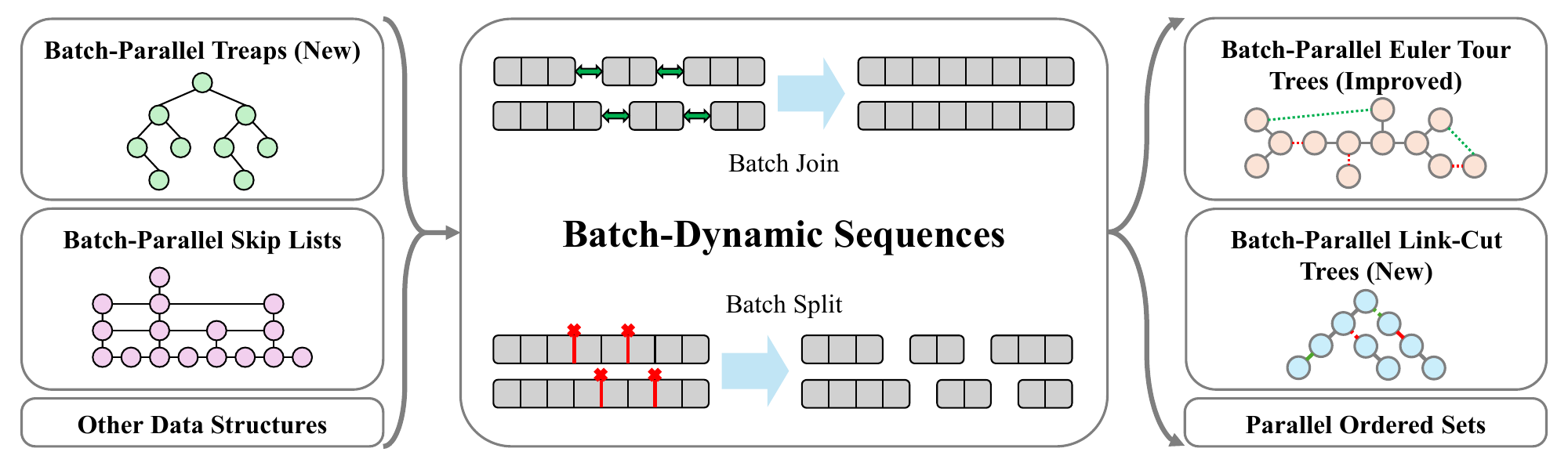}
    \vspace{-1em}
    \caption{\small Overview of the MOJOS framework. Arrows denote ``can be used to implement'' relationships.}
    \label{fig:intro_figure}
\end{figure*}

A key limitation of existing batch-parallel ETTs is their tight coupling to skip lists: Euler tours are naturally cyclic sequences, and the batch-parallel ETT algorithm of Tseng et al.~\cite{tseng2019batch} exploits the fact that skip lists can natively represent circular sequences.
This reliance is unfortunate since skip lists require $2n$ nodes in expectation to store $n$ elements, roughly doubling the memory footprint compared to tree-based alternatives.
Furthermore, in the sequential setting, ETTs can be implemented using any balanced binary search tree, and in practice, implementations based on treaps~\cite{seidel1996randomized} and splay trees~\cite{sleator1985self} significantly outperform those based on skip lists.
This raises two natural questions:

\begin{question}\label{q:treaps}
Can we design batch-parallel treap algorithms that match the theoretical bounds of skip lists and translate the sequential advantages of treaps (faster operations and lower memory usage) into the parallel setting?
\end{question}

\noindent If treaps can be made competitive or faster than skip lists in parallel, we might hope to eliminate the dependence on skip lists altogether:

\begin{question}\label{q:decouple}
Can batch-parallel ETTs be decoupled from skip lists entirely, so that any efficient batch-dynamic sequence data structure can serve as a drop-in replacement?
\end{question}

Questions~\ref{q:treaps} and~\ref{q:decouple} concern subtree queries, for which ETTs are the method of choice.
For path queries, the situation is arguably worse.
In the sequential setting, link-cut trees (LCTs)~\cite{sleator1983data, sleator1985self} have been the gold standard for over 40~years, offering $O(\log n)$ amortized time updates with small constant factors.
No existing batch-parallel tree data structure, including their sequential counterparts, comes close to matching the single-threaded performance of LCTs.
Despite this, developing a theoretically-efficient and practical batch-parallel LCT has remained a major open problem.
A key difficulty is that the sequential LCT algorithm relies on splay trees~\cite{sleator1985self}, whose top-down restructuring appears difficult to parallelize; indeed, prior work on parallel self-adjusting search structures~\cite{agrawal2018parallel} avoids splay trees entirely.
Moreover, no existing batch-dynamic tree data structure supporting path queries achieves $O(\log n)$ depth for batch updates in the binary-forking model~\cite{binaryforking}.

\begin{question}\label{q:lct}
Can we design a theoretically-efficient batch-parallel link-cut tree? Moreover, can such a data structure outperform existing parallel batch-dynamic trees for path queries in practice?
\end{question}

In this paper, we resolve Questions~\ref{q:treaps}--\ref{q:lct} above by developing a framework for constructing batch-dynamic trees from batch-dynamic sequences, and using it together with new batch-parallel treap algorithms to obtain faster batch-parallel ETTs and the first batch-parallel link-cut trees.
It has long been known that both Euler tour trees and link-cut trees can be implemented sequentially using \emph{dynamic sequence} data structures that support splitting and joining ordered sequences of elements~\cite{sleator1983data, henzinger1995randomized}.
However, this modularity has not been exploited in the parallel setting, where ETT implementations are tightly coupled to skip lists and no batch-parallel LCT algorithm exists.
We introduce the MOJOS (\textbf{M}ulti-\textbf{O}bject \textbf{J}oin \textbf{O}r \textbf{S}plit) framework (Figure~\ref{fig:intro_figure}), which formalizes this shared abstraction as an interface for batch-dynamic sequences and provides batch-parallel ETT and LCT algorithms expressed purely in terms of this interface.
As a result, any data structure implementing the batch-dynamic sequence interface can be used as a drop-in replacement to obtain batch-parallel ETTs and LCTs.

In MOJOS, we introduce a new batch-parallel treap that is both theoretically-efficient and substantially faster than skip lists in practice, answering Question~\ref{q:treaps}.
We achieve this by developing new parallel algorithms for batch joins and batch splits in treaps~\cite{seidel1996randomized} that run in $O(k \log(1+n/k))$ \revision{expected} work and $O(\log n)$ depth with high probability (\whp{}) for a batch of $k$ operations on sequences with $n$ total elements, matching the bounds of batch-parallel skip lists~\cite{tseng2019batch}.
Our algorithms run each join or split concurrently bottom-up and synchronize only through atomic compare-and-swap instructions, avoiding the synchronous top-down approach taken in prior treap algorithms~\cite{blelloch2016just, tseng2019batch}.
Our algorithms extend naturally to zip trees~\cite{tarjan2021zip} and zipzip trees~\cite{gila2023zipzip}.
In our experiments, our batch-parallel treap is $3.4\times$ faster than batch-parallel skip lists on batch joins and $1.6\times$ faster on batch splits, while using $33\%$ less memory.

Building on dynamic sequences, we design a new batch-parallel ETT algorithm that works with \emph{any} batch-dynamic sequence, answering Question~\ref{q:decouple}.
Adapting parallel ETT algorithms to the MOJOS interface requires new techniques for handling acyclicity and a parallel cycle-detection step to prevent the batch join from creating invalid Euler tours.
We introduce two optimizations that eliminate this overhead and reduce memory usage.
Our treap-based ETT is up to $1.6\times$ faster (and $1.2\times$ faster on average) for updates than the skip-list-based ETT~\cite{tseng2019batch}, while using $24\%$ less memory.

Finally, by building on dynamic sequences in MOJOS, we design the first theoretically-efficient batch-parallel link-cut tree, answering Question~\ref{q:lct}.
Our algorithm achieves $O(k \log(1+n/k))$ expected work and $O(\log n)$ depth \whp{}, and is the first batch-dynamic tree data structure supporting path queries to achieve $O(\log n)$ depth in the binary-forking model.
We present two variants.
The first is a simple algorithm that directly leverages the MOJOS framework with our batch-parallel treap and is extremely fast in practice on non-adversarial inputs.
The second is a robust algorithm that maintains a \emph{heavy preferred child invariant}, ensuring worst-case efficiency regardless of input structure.
In our experiments, both variants significantly outperform UFO trees~\cite{deman2026ufo}, the prior state-of-the-art for batch-parallel path queries.

In summary, our main contributions are as follows:
\begin{itemize}[leftmargin=*,topsep=0pt,itemsep=0pt]
\item \textbf{Batch-parallel treaps.} We develop new practical algorithms for batch joins and batch splits in treaps, achieving $O(k \log(1+n/k))$ \revision{expected} work and $O(\log n)$ depth \whp{}. Our algorithms extend naturally to zip trees and zipzip trees. Experimentally, our treap is $3.4\times$ faster than batch-parallel skip lists on batch joins, $1.6\times$ faster on batch splits, and uses $33\%$ less memory.

\item \textbf{MOJOS framework and batch-parallel ETTs.} We introduce the MOJOS framework, which decouples batch-parallel ETTs and LCTs from any specific sequence data structure. Using treaps, our ETT is up to $1.6\times$ faster (and $1.2\times$ faster on average) than the skip-list-based ETT, while using $24\%$ less memory.

\item \textbf{Batch-parallel link-cut trees.} We present the first theoretically efficient batch-parallel link-cut tree, achieving $O(k \log(1+n/k))$ expected work and $O(\log n)$ depth \whp{}--- the first $O(\log n)$-depth batch-dynamic tree for path queries in the binary-forking model. Both our simple and robust variants outperform UFO trees, the prior state-of-the-art for path queries.
\end{itemize}

\section{Preliminaries and Related Work}

\myparagraph{Model}
We use the binary fork-join model for parallelism~\cite{binaryforking,CLRS}.
A process can \defn{fork} two child threads to work in parallel, and when both complete, the parent continues.
The \defn{work} of an algorithm is the total number of instructions and the \defn{depth} (span) is the longest sequence of dependent instructions.
A computation with $W$ work and $D$ depth can be executed in $W/P + O(D)$ time with high probability on $P$ processors using a randomized work-stealing scheduler~\cite{BL98,ABP01}.
We use $O(f(n))$ \emph{with high probability} (\whp{}) to mean $O(cf(n))$ with probability $\geq 1-n^{-c}$ for $c \geq 1$.

We consider an algorithm to be \defn{phase-concurrent} if operations of the same type can be executed concurrently such that, upon the phase's completion, the resulting state is equivalent to some valid sequential execution of those operations. This definition only guarantees correctness at phase boundaries, as opposed to definitions in prior works which require linearizability~\cite{shun2014phase}.

We use standard atomic memory operations. $\oname{AtomicLoad}$ atomically reads the current value of a memory location.
$\oname{AtomicStore}$ unconditionally writes a new value.
$\oname{AtomicExchange}$ writes a new value while simultaneously returning the previous contents.
Atomic compare-and-swap, \revision{denoted as $\oname{CAS}(\vname{address}, \vname{expected}, \vname{new})$}, conditionally updates a memory location to a $\vname{new}$ value only if its current contents match an $\vname{expected}$ value, and returns whether is succeeded.

\myparagraph{Parallel Primitives}
Our algorithms rely on several standard parallel primitives. \defn{Prefix sum} computes the cumulative sums of an array's elements based on an associative operator~\cite{Blelloch89}. \defn{Filter} extracts elements that satisfy a boolean predicate into a new contiguous array, preserving their original relative order~\cite{blelloch1990vector}. \defn{Semisort} groups an array's elements such that those with identical keys become contiguous, though the ordering between distinct keys is arbitrary~\cite{gu2015top}. Finally, \defn{list contraction} reduces a set of linked lists or simple cycles by repeatedly merging adjacent nodes~\cite{blelloch2019optimal}. For an input of size $n$, each of these primitives operates within $O(n)$ work and $O(\log n)$ depth; for semisort and list contraction, the work bounds are in expectation and the depth bounds hold with high probability.

\subsection{Dynamic Sequences and Dynamic Trees}
The \defn{dynamic sequences problem} is to maintain a data structure representing a collection of ordered sequences of elements, subject to joining and splitting sequences, and various query operations (see Section~\ref{sec:mojos} for more details).
Dynamic sequences are efficiently implementable using structures like balanced binary search trees, splay trees~\cite{sleator1985self}, skip lists~\cite{pugh90}, and treaps~\cite{seidel1996randomized}. Furthermore, they serve as building blocks for more complex data structures, including ordered sets, Euler tour trees, and link-cut trees.
Batch-dynamic sequence data structures such as batch-parallel skip lists~\cite{tseng2019batch} allow for batch-parallel joins, splits, and queries.
Augmented dynamic sequences associate a value from a domain $D$ with each element and maintain internal aggregates over some subsequences using an associative and commutative function $f: D \times D \to D$, enabling efficient operations such as range aggregate queries.

The \defn{dynamic trees problem}~\cite{sleator1983data} is to maintain a data structure representing a forest while supporting efficient edge insertions ($\link$), edge deletions ($\cut$), and vertex connectivity queries.
All dynamic tree data structures can support connectivity queries, and many also support more advanced queries.
\emph{Path queries} and \emph{subtree queries} compute a commutative and associative function applied over the weights along a path or within a subtree, respectively.
Link-cut trees~\cite{sleator1983data,sleator1985self} support path queries and perform operations in $O(\log n)$ amortized or worst-case time.
Euler tour trees~\cite{henzinger1995randomized} support subtree queries and perform operations in $O(\log n)$ time \whp{}.

The \defn{batch-dynamic trees problem}~\cite{tseng2019batch} supports batches of edge insertions ($\batchlink$), edge deletions ($\batchcut$), and batch queries.
%
Batch-parallel Euler tour trees~\cite{tseng2019batch} are the current state-of-the-art for batch-dynamic trees with subtree queries.
UFO trees~\cite{deman2026ufo} are the current state-of-the-art for path queries and also support many other types of queries.
Prior work on concurrent link-cut trees~\cite{stoian2022concurrent} achieves some parallelism within each individual operation, but it does not prove any improved theoretical bounds or consider batch-parallel operations.

\section{The MOJOS Framework} \label{sec:mojos}

In this paper we introduce the \textbf{M}ulti-\textbf{O}bject \textbf{J}oin \textbf{O}r \textbf{S}plit (\defn{MOJOS}) framework, and we define the abstract dynamic sequence and batch-dynamic sequence data structures.
The MOJOS framework shows how to use a simple yet powerful \defn{batch-dynamic sequence} data structure to implement more complex parallel data structures, such as ordered sets,  Euler tour trees, and link-cut trees.

The (sequential) \defn{dynamic sequences problem} is to maintain a data structure representing a collection of ordered sequences of elements, providing the following operations:
\begin{itemize}[leftmargin=*,topsep=0pt,itemsep=0pt]
    \item $\myjoin(x,y)$: Concatenates two distinct sequences by joining the last element $x$ of one sequence to the first element $y$ of another.
    \item $\mysplit(x,y)$: Splits a sequence between consecutive elements $x$ and $y$ into two distinct sequences, preserving order.
    \item $\getrep(x)$: Returns a unique identifier for the sequence containing $x$ (e.g., the root node in a tree). \revision{This can be any element in the sequence, as long is the return value is consistent for all elements in the same sequence.}
    \item $\getpred(x) / \getsucc(x)$: Returns the immediate predecessor or successor of $x$ in its sequence.
    \item $\gethead(x) / \gettail(x)$: Returns the first or last element of the sequence containing $x$.
\end{itemize}
While existing literature typically defines $\mysplit$ with a single element (splitting immediately to its right), we adopt a two-element interface. This more general formulation accommodates implementations that require explicit references to both adjacent elements, such as our batch-parallel treap algorithm.

The \defn{batch-dynamic sequences problem} supports batch joins and splits across multiple sequences (hence the name MOJOS). It implements the following additional batch operations:
\begin{itemize}[leftmargin=*,topsep=0pt,itemsep=0pt]
    \item $\batchjoin(U) / \batchsplit(U)$
    \item $\batchgetrep(Q)$
    \item $\batchgetpred(Q) / \batchgetsucc(Q)$
    \item $\batchgethead(Q) / \batchgettail(Q)$
\end{itemize}
These methods accept either a list of update pairs $U$ or a list of query elements $Q$.
We assume that no input lists contain duplicates, the inputs to batch-join don't induce cycles, and the inputs to batch split are only pairs of currently adjacent elements.
Batch updates ($\batchjoin$ and $\batchsplit$) must yield a state equivalent to executing the updates in $U$ sequentially in any valid order.
Batch queries return a list of results corresponding to the elements in $Q$.

\myparagraph{MOJOS Framework}
Sequential dynamic sequences are well studied and are efficiently implementable using structures like balanced binary search trees, splay trees~\cite{sleator1985self}, skip lists~\cite{pugh90}, and treaps~\cite{seidel1996randomized}. 
They serve as versatile building blocks for more complex data structures, including ordered sets, Euler tour trees (ETTs), and link-cut trees (LCTs). 
However, this modularity offered by dynamic sequences is largely unexplored in the parallel setting. 
For instance, existing batch-parallel ETT algorithms are tightly coupled to skip lists, and there is no known batch-parallel algorithm for LCTs.

To bridge this gap, we introduce the MOJOS framework, which provides efficient batch-parallel algorithms for ETTs and LCTs in terms of batch-dynamic sequences. The primary advantage of MOJOS is this decoupling: any underlying sequence data structure supporting this fairly minimal interface can be used to build batch-parallel ETTs and LCTs.
Our framework also accommodates applications requiring augmented values maintenance in the dynamic sequence data structure (although we do not list these additional operations in our core interface for simplicity).

\myparagraph{MOJOS for Ordered Sets}
The standard MOJOS interface maintains sequence order purely through explicit split and join operations, without requiring intrinsic element keys. However, applications like ordered sets naturally dictate sequence order based on keys. To support such applications, the underlying data structure must implement an additional search operation:
\begin{itemize}[leftmargin=15pt,topsep=0pt,itemsep=0pt]
    \item $\find(x) / \batchfind(Q)$: Returns the element with the maximum key less than or equal to $x$ in a sequence. The batch variant receives a pre-sorted list of elements.
\end{itemize}
While not the primary focus of this paper, we detail simple parallel ordered set algorithms using MOJOS in \appref{Appendix~\ref{app:set}} which achieves both optimal work and optimal depth (other algorithms are either more complex or have suboptimal depth). Our experiments in \appref{Appendix~\ref{app:set}} demonstrate that an implementation of this approach achieves reasonable performance compared to highly optimized state-of-the-art implementations.

\section{Batch-Parallel Treaps}

In this section we describe our new parallel algorithms for batch joins and batch splits in treaps.
Our data structure supports the entire batch-dynamic sequence interface and thus fits into the MOJOS framework.
Our algorithms also extend naturally to augmented treaps and other randomized binary search tree data structures such as zip trees, zipzip trees, and
biased zipzip trees. We discuss these extensions at the end of this
section.

\revision{
\myparagraph{Review of Treaps}
A treap~\cite{seidel1996randomized} is a randomized binary search tree data structure representing each element in a dynamic sequence with a treap node.
Each treap node contains the following fields:
\begin{itemize}[leftmargin=*,topsep=0pt,itemsep=0pt]
    \item $\parent$: the parent of this node.
    \item $\child[0]$: the left child of this node. 
    \item $\child[1]$: the right child of this node.
    \item $\priority$: the random priority of this node.
\end{itemize}
Treaps maintain the invariant that for each node its priority is the largest in its subtree, and the in-order traversal of the BST defines the dynamic sequence order.
Thus for two successive nodes in the treap, the one with higher priority is an ancestor of the other.
Sequentially, joining two treaps is done by interleaving the spines of the two elements in increasing priority order.
Splitting a treap essentially separates a single spine into the two sides of the split while maintaining the priority order.
Full descriptions of the sequential algorithms are provided in \appref{Appendix~\ref{app:treap_seq}}.

Our algorithms also use a function \defn{$\parentdirection$}, which returns $0$ or $1$ depending on if this node is the left or right child of its parent respectively.
This can be implemented by using one bit in the parent pointer of the current node.
This is helpful when the parent node's child pointer has been changed during the update, and just checking the children of the parent doesn't work.
}

\subsection{Treap Batch Join and Batch Split}
Our algorithms for batch join and batch split take $O(k \log (1+n/k))$ work and $O(\log n)$ depth both \whp{}. We prove the following theorem in \appref{Appendix~\ref{app:treap_cost}}:

\begin{restatable}{theorem}{batchtreaptheorem}\label{thm:batchjoin}
    The batch join and batch split algorithms for treaps take $O(k \log (1+n/k))$ \revision{expected} work and $O(\log n)$ depth with high probability, where $k$ is the size of the batch and $n$ is the total number of elements in all treaps involved in the batch of joins.
\end{restatable}

Our algorithms are both based on simple sequential bottom-up treap join and split algorithms which we describe in detail in \appref{Appendix~\ref{app:treap_seq}}. We found that bottom-up algorithms are much more amenable to concurrency and have better practical performance than typical sequential top-down algorithms for treaps (see Section~\ref{sec:seq_exp}).
Designed for minimal synchronization, our batch join algorithm is phase-concurrent: each join runs locally, synchronizing only through atomic CAS operations. 
Our batch split algorithm runs in two phase-concurrent steps: the first safely modifies child pointers to structurally split the treaps, and the second resolves the parent pointers to ensure consistency.

\myparagraph{Phase-Concurrent Joins}
Algorithm~\ref{alg:concurrentjoin} shows the pseudo-code for phase-concurrent joins. $\batchjoin$ simply calls $\fname{ConcurrentJoin}$ in parallel for each pair in the input batch.

\begin{algorithm}[ht]
\small
\caption{$\fname{ConcurrentJoin}(x,y)$}
\label{alg:concurrentjoin}
$X = x$, $Y = y$\;
\While{$X$ \textbf{and} $Y$}{
    \If{$X.\priority < Y.\priority$}{
        \While{$X.\parent$ \textbf{and} $X.\parent.\priority < Y.\priority$ \label{line:symmetric_start}}{
            $X = X.\parent$\;
        }
        $next_X = X.\parent$\;
        $Y.\child[0] = X$\label{line:join_x_start}\;
        \If{$\oname{CAS}(\&X.\parent, next_X, Y)$}{
            $X = next_X$\;
        }\label{line:join_x_end}\label{line:symmetric_end}
    }
    \lElse{
        \textit{(Symmetric to lines~\ref{line:symmetric_start}--\ref{line:symmetric_end})}
    }
}
\end{algorithm}


\begin{figure}[t]
    \centering
    \includegraphics[width=\linewidth]{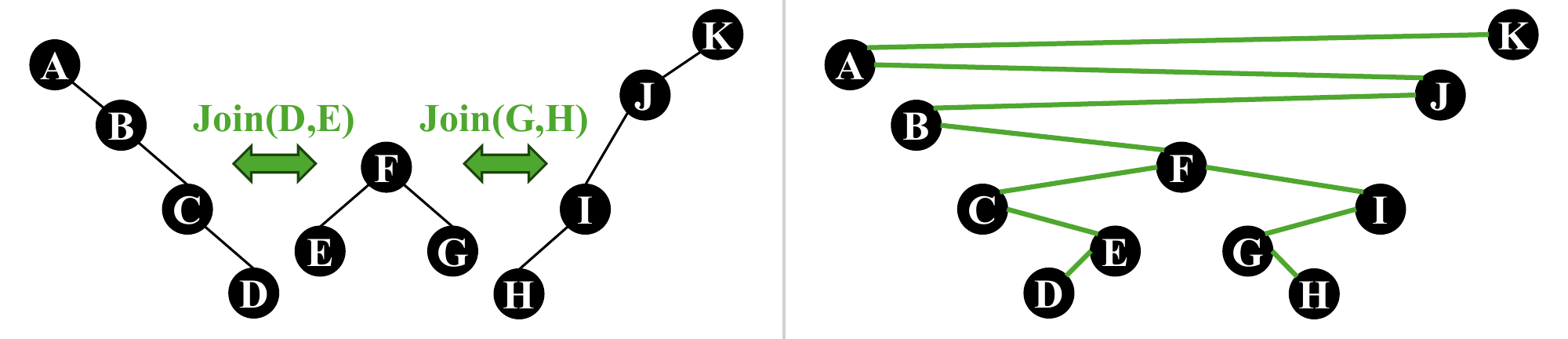}
    \caption{\small An example of a batch join of 3 treaps. The right shows the resulting treap with new pointers in green.}
    \label{fig:batch_join}
\end{figure}

Joining two treaps essentially entails interleaving their right and left spines bottom-up in increasing priority order. In a batch, spines being merged are almost disjoint, conflicting only at root nodes involved in multiple joins. 
The main difference from the sequential algorithm is the CAS operation on the parent pointer (lines~\ref{line:join_x_start}--\ref{line:join_x_end}). There is only contention on this CAS if the node is currently the root of a treap ($next_x = \vname{null}$).
If this is not the case, the CAS always succeeds and the join proceeds as normal.
A successful CAS on a root's parent pointer completes the join, and the while loop will break on the next iteration. If it fails, another join has already modified the parent, so the current join restarts from the current node, finding the new parent as it traverses up the spine.

Consider element $F$ in Figure~\ref{fig:batch_join} as an example. The thread for $\myjoin(D,E)$ would attempt to CAS the parent pointer of $F$ to $B$, while the thread for $\myjoin(G,H)$ would attempt to CAS it to $J$.
If the CAS to $B$ succeeded first, $\myjoin(D,E)$ completes, and $\myjoin(G,H)$ will interleave the remaining elements above $F$.
If the CAS to $J$ succeeded first, $\myjoin(G,H)$ completes, and $\myjoin(D,E)$ will interleave the remaining elements above $F$, including interleaving element $B$ between $F$ and $J$ on the combined spine.

Crucially, in our algorithm the child pointer is updated before the corresponding parent pointer CAS (line~\ref{line:join_x_start}).
If updated after, a thread could halt after the CAS, allowing a second join to modify the parent again. Once the first thread resumes, it would overwrite the child pointer erroneously. Setting the child beforehand ensures structural integrity because even if the CAS fails, the join restarts and eventually sets the correct child value. Contention on the valid child pointer is impossible, as it would imply two separate joins are merging into the same spine.

\myparagraph{Two-Phase Batch Splits}
Algorithm~\ref{alg:concbatchsplit} shows the pseudo-code for $\batchsplit$, which runs in two phases. The first phase runs $\fname{ConcurrentSplit}$ (Algorithm~\ref{alg:concurrentsplit}) to change necessary child pointers. The second phase (Algorithm~\ref{alg:splitphase2}) corrects the parent pointers.


\begin{algorithm}[ht]
\small
\caption{$\fname{ConcurrentSplit}(x, y)$}
\label{alg:concurrentsplit}
$X = x$, $Y = y$\;

\If{$X.\priority > Y.\priority$}{
    \lWhile{$Y.\parent \neq X$}{ $Y = Y.\parent$ \label{line:split_init_start}}
    $c = \oname{AtomicExchange}(\&X.child[1], \vname{null})$ \label{line:unset_child}\;
    \lIf{pointer $c$ is marked}{ \Return \label{line:split_init_end}}
}
\lElse{
    \textit{(Symmetric to lines~\ref{line:split_init_start}--\ref{line:split_init_end})}
}

\While{$X$ \textbf{and} $Y$}{
    \If{$X.\priority > Y.\priority$}{
        $P = X.\parent$ \label{line:split_main_start}\;
        \If{$P$ \textbf{and} $X.\parentdirection == 0$ \label{line:split_dir_change} }{
            $curr = \oname{AtomicLoad}(\&P.child[0])$ \label{line:child_load}\;
            \If{$Y.\priority < curr.\priority$ \label{line:prio_check}}{
                Let $Y^*$ be a marked pointer to $Y$.\;
                \If{$\oname{CAS}(\&P.child[0], curr, Y^*)$ \label{line:child_CAS}}{
                    \lIf{pointer $curr$ is marked}{ \Return \label{line:split_early_terminate} }
                    \lElse{ $Y = P$ }
                }
            } \lElse { \Return \label{line:prio_check_fail} }
        }
        \lElse { $X = P$ \label{line:split_main_end} }
    }
    \lElse{
        \textit{(Symmetric to lines~\ref{line:split_main_start}--\ref{line:split_main_end}})
    }
}
\end{algorithm}

\begin{algorithm}[ht]
\small
\caption{$\fname{CleanUpParentPointer}(P, C, dir)$}
\label{alg:splitphase2}
    $curr^* = \oname{AtomicLoad}(\&P.child[dir])$ \;
    Let $curr$ be an unmarked copy of $curr^*$. \;
    \If{$curr == C$}{
        $C.\parent = P$\;
        $\oname{AtomicStore}(\&P.child[dir], C)$
    }
    \lElse{
        $C.\parent = \vname{null}$
    }
\end{algorithm}

\begin{algorithm}[ht]
\small
\caption{$\fname{BatchSplit}(splits)$}
\label{alg:concbatchsplit}
\parfor{$(x,y) \in splits$} {
    $\fname{ConcurrentSplit}(x,y)$\;
}
\parfor{attempted child pointer changes} {
    $\fname{CleanUpParentPointer}(parent, child, dir)$\;
}
\end{algorithm}

\myparagraph{Concurrent Split Phase}
\revision{This first phase only updates child pointers and leaves all parent pointers completely unmodified.}
Since parent pointers are unmodified in this phase, concurrent upward traversals can safely read the original tree structure.
\revision{As described previously, we use $\parentdirection$ to determine whether a node is the left or right child of its parent since we cannot safely check the parent's child pointers.}

Each split traverses bottom-up from the higher priority element, distributing nodes to the left or right side of the split based on path direction changes.
Consider the split $(G,H)$ in Figure~\ref{fig:batch_split}. \revision{Since $G$ has higher priority than $H$,} the algorithm first atomically clears $G$'s right child pointer (line~\ref{line:unset_child}) and then splits bottom-up from $G$.
Additional child pointer changes occur only when the upward path changes direction (line~\ref{line:split_dir_change}). Moving from $G$ to $K$, the direction changes from left to right, so the most recent right-side node, $J$, becomes the left child of $K$. Subsequently, when changing from right to left, $G$ becomes the right child of $F$. Our algorithm handles this by atomically setting the child pointer of the parent only if its priority is lower than the current child pointer (lines~\ref{line:prio_check}--\ref{line:child_CAS}).

When multiple splits occur concurrently, such as $(B,C)$, $(C,D)$, and $(D,E)$ in Figure~\ref{fig:batch_split}, they may compete to update the same child pointer. Uncoordinated, all three would attempt to set $C$, $D$, and $E$ as the left child of $F$, since $F$ is their next direction-change node. The correct left child is $E$ (the rightmost node). Generally, the correct child will always be the one of minimum priority among competitors. Our algorithm enforces this by atomically loading the current child (line~\ref{line:child_load}) and only attempting a CAS if the new node has a lower priority (lines~\ref{line:prio_check}--\ref{line:child_CAS}).
\revision{If the minimum priority node fails the CAS due to a contending operation, it will simply retry the main while loop, and succeed on a later iteration.}
If the current node's priority is higher than the current child, another split has already won, and the thread terminates (line~\ref{line:prio_check_fail}).
To ensure work-efficiency, successfully changed pointers are marked \revision{(by setting a bit in the pointer)}. Any thread that successfully overwrites an already marked pointer terminates (line~\ref{line:split_early_terminate}), preventing redundant work along overlapping split paths.
Despite some overlap and some CAS contention, we prove in \appref{Appendix~\ref{app:treap_cost}} that the total work in this phase is $O(k \log (1+n/k))$ \whp{}.

\begin{figure}[t]
    \centering
    \includegraphics[width=\linewidth]{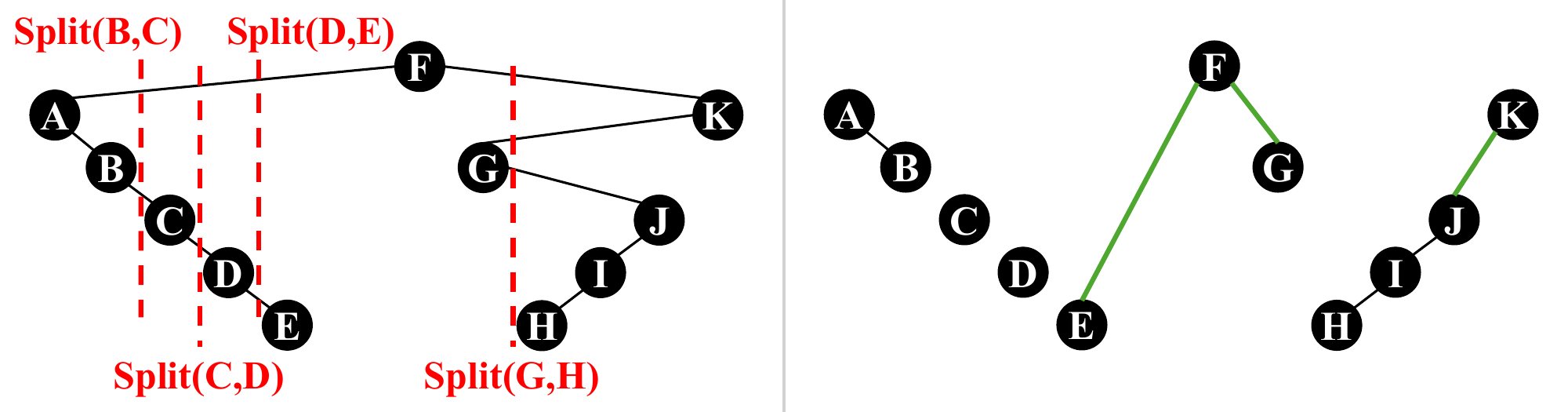}
    \caption{\small
    An example of a batch split in a treap. The right shows the resulting treaps with new pointers in green.}
    \label{fig:batch_split}
\end{figure}

\myparagraph{Cleanup Phase}
The second phase (Algorithm~\ref{alg:splitphase2}) resolves the corresponding parent pointers for changed child pointers from the first phase.
\revision{During the first phase, our algorithm tracks each attempt to set an element $C$ as the direction $dir$ child of an element $P$.
For each such attempt, the second phase} checks if $P$'s child pointer actually points to $C$.
If so, it sets the parent pointer of $C$ to be $P$\revision{, and replaces $P$'s child pointer with an unmarked copy.}
Otherwise, it clears the parent pointer, correctly identifying the node as a new treap root.
Nodes that were the top of their spine when the split loop broke\revision{, such as $A$ and $K$ in Figure~\ref{fig:batch_split},} also clear their parent pointers in this phase \revision{(not shown in pseudo-code)}.

\subsection{Extensions of Treap Algorithms}

We show in \appref{Appendix~\ref{app:treap_query}} how to implement all MOJOS batch query operations in $O(k \log (1+n/k))$ \revision{expected} work and $O(\log n)$ depth \whp{}.
In \appref{Appendix~\ref{app:augmented_treap}} we also describe how to adapt common techniques for maintaining augmented values~\cite{tseng2019batch} to our batch-parallel treap batch algorithms.

\myparagraph{Zip, Zipzip, and Biased Zipzip Trees}
Zip trees~\cite{tarjan2021zip}, zipzip trees~\cite{gila2023zipzip}, and biased zipzip trees~\cite{gila2023zipzip} share a similar structure to treaps but use different random distributions for element priorities. Our batch-parallel algorithms directly extend to these data structures with minimal modification.
The extension to biased zipzip trees is particularly useful for efficiently implementing link-cut trees. Biased zipzip trees maintain a weight $w(x)$ for each element $x$, and the depth of $x$ is $O(\log (W / w(x)))$ \whp{}, where $W$ is the total weight. In our use cases $W \leq n$, meaning our algorithms retain $O(k \log (1+n/k))$ \revision{expected} work and $O(\log n)$ \whp{} depth bounds. In \appref{Appendix~\ref{app:zipzip}}, we describe these algorithmic extensions as well as a $\batchupdate(U)$ function to update weights for a batch of elements within these same bounds.

\section{Batch-Parallel Euler Tour Trees}

Sequential Euler tour trees (ETTs) are well known to be implementable using various dynamic sequence data structures, including skip lists~\cite{pugh90}, treaps~\cite{seidel1996randomized}, splay trees~\cite{sleator1985self}, and balanced binary search trees.
However, existing batch-parallel ETT algorithms~\cite{tseng2019batch} specifically rely on skip lists and their unique ability to represent cyclic sequences.
Recent works~\cite{tseng2019batch,deman2026ufo} demonstrate that sequential ETTs built on alternative structures (particularly splay trees and treaps) often yield faster updates, faster queries, and lower memory usage. Furthermore, our empirical evaluation shows that batch-parallel treaps consistently outperform their skip list counterparts.

Motivated by these advantages, we aim to design a generalized batch-parallel ETT algorithm that does not depend on any specific underlying data structure.
In this section, we achieve this by adapting existing algorithms to the MOJOS framework, enabling the use of any batch-dynamic sequence data structure.
Finally, we introduce treap-specific optimizations that further enhance performance.

\revision{
\myparagraph{Review of ETTs}
An ETT~\cite{henzinger1995randomized} represents the Euler tours of a forest using a dynamic sequence data structure. Although various Euler tour representations exist, we use edge-based Euler tours which have two elements per undirected edge (one for each direction), as well as a single ``self-loop'' element per vertex.
Since each Euler tour forms a cycle, the dynamic sequence data structure represents the cycle as a sequence by picking some arbitrary starting point for each tour (if using skip lists the cycles can be represented directly).
The ETT maintains an array of the dynamic sequence elements corresponding to each vertex's self-loop element, indexed by vertex id.
It also maintains a hash table mapping edges to their two corresponding dynamic sequence elements.

A sequential $\link$ operation merges two tours by splicing their sequences at the self-loop elements of the two newly connected vertices.
A sequential $\cut$ operation splits a tour into two by removing the two directed elements of the deleted edge and re-stitching the remaining sequences. Each operations takes a constant number of $\myjoin$ and $\mysplit$ operations in the dynamic sequence data structure.
}

\subsection{Parallel ETT in the MOJOS Framework}

Existing batch-parallel ETT algorithms natively rely on cyclic skip lists with direct predecessor/successor pointers and fully phase-concurrent joins and splits~\cite{tseng2019batch}. 
Additionally, their split operation requires only a single element (splitting immediately to its right). 
In contrast, the MOJOS framework generalizes to non-cyclic sequences and requires both left and right elements as the input to splits. 
We adapt the batch link algorithm for ETTs to this more general interface below.
We defer the analogous batch cut algorithm to \appref{Appendix~\ref{app:ett_cut}}, as the techniques are similar.
Our algorithms also support augmented ETTs with subtree queries (assuming that the batch-dynamic sequence supports augmented values) which we discuss further in \appref{Appendix~\ref{app:augmented_ett}}.
Finally, we prove the following theorem about the costs of batch updates in \appref{Appendix~\ref{app:ett_analysis}}.

\begin{restatable}{theorem}{ettanalysis} \label{thm:ett_analysis}
    Given a batch-dynamic sequence data structure that can perform batches of links, cuts, and queries in $O(k \log(1+n/k))$ \revision{expected} work and $O(\log n)$ depth \whp{} where $k$ is the batch size, the \text{MOJOS} Euler tour tree batch link and batch cut algorithms take $O(k \log(1+n/k))$ expected work and $O(\log n)$ depth \whp{}.
\end{restatable}

\revision{
\myparagraph{Existing Algorithm Summary}
We briefly summarize the existing ETT batch link algorithm~\cite{tseng2019batch}.
The algorithm maintains each Euler tour cycle directly with a skip list.
The batch link procedure first splits the skip list after the self-loop element corresponding to each vertex with an incident link.
This fragments the existing Euler tours and opens up insertion points to join the new edges in and splice the Euler tours together.

Next, the algorithm groups new directed edges by their source endpoint and arbitrarily orders each group to establish a strict cyclic ordering of the new edges locally around each vertex.
Guided by these local orderings the algorithm can independently determine the skip list joins needed in the neighborhood of each vertex.
That is, each vertex self-loop joins to the first outward edge incident to it, each inward edge joins to the next outward edge, and the last outward edge joins to the previous successor of the vertex self-loop.
The authors proved that this set of joins correctly joins the Euler tour fragments into cycles.
}

\begin{algorithm}[ht]
\small
\caption{$\fname{BatchLink}(\vname{links})$}
\label{alg:ett_batch_link}
    $\vname{dir\_edges} =$ each edge from $\vname{links}$ directed both ways. \;
    \parfor{$e \in \vname{dir\_edges}$ \label{line:init_start}}{
        Instantiate an element for $e$. \label{line:init_end}
    }
    $\vname{splits}[0 \dots \vname{links}.\vname{size}()-1] = \vname{null}$ \;
    $\vname{joins}[0 \dots 3\text{$\cdot$}\vname{links}.\vname{size}()-1] = \vname{null}$
    
    \parfor{$(u,v) \in \vname{links}$ \label{line:collect_splits_start}}{
        
        \For{$w \in (u,v)=\vname{links}[i]$}{
            \If{$\getsucc(w) \neq \vname{null}$}{
                $\vname{splits}[i] = (w, \getsucc(w))$ \label{line:add_split}\;
                \If{$\gettail(w)$ is not a vertex in $\vname{links}$}{
                    $\vname{joins}[3i] = (\gettail(w), \gethead(w))$ \label{line:tail_head_join}
                }
            }
        }
    }
    $\vname{sorted\_edges} = \fname{Semisort}(\vname{dir\_edges})$ \label{line:semisort}
    \parfor{$e = (u,v) \in \vname{sorted\_edges}$ \label{line:join_loop_start}}{
        \lIf{$e$ is the first edge out of $u$}{$\vname{joins}[3i+1] = (u,e)$}
        $e' = (v,u)$, $u^+ = \getsucc(u)$ \;
        $e^+ = \vname{sorted\_edges}[i+1]$ \;
        \lIf{$e^+.\vname{src} == u$}{$\vname{joins}[3i+2] = (e',e^+)$}
        \lElse{$\vname{joins}[3i+2] = (e', u^+)$}
    \label{line:join_loop_end}}
    $\batchsplit(\fname{Clean}(\vname{splits}))$ \label{line:batch_split}\;
    $\batchjoin(\fname{RemoveCycles}(\fname{Clean}(\vname{joins})))$ \label{line:batch_join}
\end{algorithm}

\myparagraph{Generalized Batch Link}
\revision{Algorithm~\ref{alg:ett_batch_link} details our ETT batch link procedure, which generalizes the existing algorithm~\cite{tseng2019batch}.
The batch-parallel ETT uses a batch-dynamic sequence to represent the tours and a batch-parallel hash table~\cite{gil1991towards} to map edges to their sequence elements.}
For each new edge, we first instantiate twin directed edge elements (lines~\ref{line:init_start}--\ref{line:init_end}) \revision{and batch-insert these entries into the hash table}.
Next, we identify split points incident to endpoint vertices in the batch of links. Rather than using phase-concurrent operations as in the prior skip list algorithm, MOJOS concurrently collects the successor ($\getsucc$) for each vertex and adds each split to a list of splits (line~\ref{line:add_split}). 
Since MOJOS does not store sequences cyclically, we must add a join between the tail and head of any affected sequence whose tail is not actively being split (line~\ref{line:tail_head_join}) to preserve the ETT structure.
To add the new edges into the existing tours, we group the directed edges by their source vertex using a semisort (line~\ref{line:semisort}). As in the original algorithm, this grouping defines the set of joins to reconnect the ETT fragments (lines~\ref{line:join_loop_start}--\ref{line:join_loop_end}).

Finally, we execute the $\batchsplit$ and $\batchjoin$ operations (lines~\ref{line:batch_split}--\ref{line:batch_join}).
\revision{$\fname{Clean}$ simply removes duplicates and null values in each list.}
Crucially, because the subsequent batch of joins may form cycles (which MOJOS strictly prohibits) we must omit one join per cycle. This is handled by $\fname{RemoveCycles}$. This function maps each join element to its representative ($\getrep$) and applies a parallel list contraction to find exactly one join per cycle. It then filters out the cycle-inducing joins before calling $\batchjoin$.

\subsection{Treap ETT Optimizations}

\myparagraph{Cyclic Treap}
Applying the MOJOS framework directly incurs overhead from the list contraction and $\getrep$ calls required to prevent cycles. To avoid this, our modified treap batch join algorithm natively tolerates cycle-inducing batches (detailed in \appref{Appendix~\ref{app:cyclic_treap}}). If a batch of joins induces a cycle, the algorithm leaves the treap in a state identical to omitting exactly one join per new component. By bypassing MOJOS's cycle prevention steps, we notably boost batch-parallel ETT performance.

\myparagraph{High Priority Vertex Elements}
An ETT contains up to $3n-2$ elements ($n$ vertices and up to $2n-2$ edges). Typically, augmenting an ETT requires storing internal aggregates in every element, inflating memory usage. We optimize this specifically for treaps by assigning strictly higher priorities to vertex elements than edge elements in the treap (e.g., setting the most significant priority bit to $1$ for vertices). Because treaps maintain a heap property based on these priorities, this invariant forces all vertex elements to cluster contiguously near the root, ensuring no edge element contains a vertex in its subtree. Consequently, we can only store aggregates in the $n$ vertex elements, substantially reducing memory usage without sacrificing update speed.
We show experimental results for the memory savings in \appref{Appendix~\ref{app:high_prio_vertex_ett}}.
\section{Batch-Parallel Link-Cut Trees}

First we review sequential link-cut trees and describe how they fit into the MOJOS framework.
Then we describe simple algorithms for batch-parallel links and batch-parallel cuts which are only efficient for inputs with low diameter.
Finally, we give more complex batch-update algorithms that, when implemented with our batch-parallel biased zipzip trees, are work-efficient and polylogarithmic depth.

\subsection{Link-Cut Trees in the MOJOS Framework}

Traditional link-cut trees~\cite{sleator1983data, sleator1985self} represent the input as a forest of rooted trees partitioned into disjoint \defn{preferred paths}, where each path is stored in an \defn{auxiliary tree} (typically a splay tree). By generalizing the auxiliary tree to be any dynamic sequence data structure, the link-cut tree naturally fits into our MOJOS framework. 

Each vertex $v$ in the \defn{input tree} corresponds to an element $v'$ in the auxiliary tree. The left-to-right order of elements in an auxiliary tree corresponds to top-to-bottom order of vertices in the preferred path it represents. 
The \defn{head} (lowest depth) and \defn{tail} (highest depth) of a path correspond to the leftmost and rightmost auxiliary tree elements. Edges not on a preferred path are explicitly stored via a non-preferred parent pointer ($\vname{nppar}$) linking the head of an auxiliary tree to its parent in the input tree.

\myparagraph{Helper Functions and Sequential Algorithms}
We can express the internal link-cut tree operations and sequential updates strictly through the abstract MOJOS sequence interface:
\begin{itemize}[leftmargin=*,topsep=0pt,itemsep=0pt]
    \item $\splice(u,v)$: Sets $v$ as $u$'s preferred child.
    \textit{If $u$ has an existing preferred child $w$, we call $\mysplit(u',w')$ and set $w.\vname{nppar} = u$. If $v$ is not $\vname{null}$, we then attach the new preferred child with $\myjoin(u',v')$.}
    \item $\expose(v)$: Consolidates the path from $v$ to the root into a single preferred path.
    \textit{It unsets $v$'s preferred child, then iteratively finds the head $h$ of the current path via $\gethead(v')$ and splices it to its non-preferred parent: $\splice(h.\vname{nppar}, h)$.}
    \item $\evert(v)$: Makes $v$ the root of its input tree.
    \textit{It calls $\expose(v)$ to make $v$ the tail of its preferred path, and then calls $\reverse(v')$ to logically make it the new head.}
    \item $\reverse(x)$: Given $x$, the representative element of a sequence, it reverses the order of that sequence's elements.
\end{itemize}
Note that $\reverse$ is an operation on the auxiliary trees not included in the standard MOJOS interface. Traditional link-cut trees use a lazy reversal bit at each splay tree node to handle this. Prior to any operation involving some element, the reversal bits are propagated top-down for each ancestor of that element. At each ancestor a swap of the left and right child pointers occurs and the reversal bits of its children are flipped. In \appref{Appendix~\ref{app:lct_reverse}}, we detail how to extend this technique to treaps and other randomized binary search trees, and how to safely integrate it with batch-parallel updates.
Sequential link-cut tree updates can be expressed as follows using these operations. $\link(u,v)$ calls $\expose(u)$, $\evert(v)$, and $\myjoin(u', v')$. $\cut(u,v)$ ensures $u$ is the parent, calls $\expose(v)$, and separates them with $\mysplit(u', v')$.

\subsection{Simple Batch-Parallel LCT Updates}

We first introduce lightweight batch link and batch cut algorithms. These establish the core mechanics for our theoretically efficient algorithms and, as our experiments demonstrate, are highly practical on inputs with low diameter or other non-adversarial structures.
Theorem~\ref{thm:simple_lct_analysis} (proved in \appref{Appendix~\ref{app:lct_cost}}) proves the theoretical efficiency of these algorithms on low diameter inputs, \revision{when implemented with a suitable choice of batch-dynamic sequence.}

\begin{restatable}{theorem}{lctbatchupdatesimple}\label{thm:simple_lct_analysis}
    The simple batch link and batch cut algorithms take $O(k \cdot D)$ work and $O(\log k + D)$ depth where $k$ is the batch size and $D$ is the diameter of the input tree. Queries take $O(D)$ time.
\end{restatable}

\myparagraph{Parallel Queries}
Traditional link-cut tree queries rely on $\expose$, which mutates the tree structure to guarantee amortized bounds. This mutation makes concurrent queries difficult. Instead, we use simple, read-only query algorithms. 
For connectivity queries, we find the root of a vertex's input tree by repeatedly calling $\gethead$ on its current auxiliary tree and following the non-preferred parent pointer upward. For path queries, we augment the auxiliary trees to support efficient prefix aggregate queries. We then apply this same bottom-up traversal to accumulate the aggregate over an entire input path. We detail these algorithms along with a more work-efficient batch connectivity query in \appref{Appendix~\ref{app:lct_query}}.

\myparagraph{Batch Link}
The input to $\batchlink$ (see Figure~\ref{fig:batch_link}) is a set of edges whose insertion does not form a cycle. We can view the current connected components as vertices and the batch of new links as edges, forming a forest of \defn{link trees}. We direct this forest into a rooted forest in parallel using the classic Euler tour technique~\cite{tarjan1985efficient}. 

The main challenge for batch links is that we cannot concurrently evert multiple vertices within the same component, as a tree can only have one root. Our directed forest naturally solves this: each non-root component $C_i$ now has exactly one outgoing parent edge $(u_i, v_i)$, and we can simply only evert the source vertex for each link, calling $\evert(u_i)$ and setting $u_i$'s non-preferred parent to $v_i$. Because each component has at most one outgoing link, it undergoes at most one evert operation. This structural guarantee allows all everts and subsequent parent pointer updates to execute independently and safely in parallel.

\myparagraph{Batch Cut}
The input to $\batchcut$ is a set of existing edges. We process them in parallel based on their edge type. For non-preferred edges (identified by checking if $u.nppar = v$ or vice versa) we simply clear the non-preferred parent pointer. For preferred edges, we cut them all using a single $\batchsplit$ operation on the auxiliary trees.

\begin{figure*}[ht]
    \centering
    \includegraphics[width=0.9\linewidth]{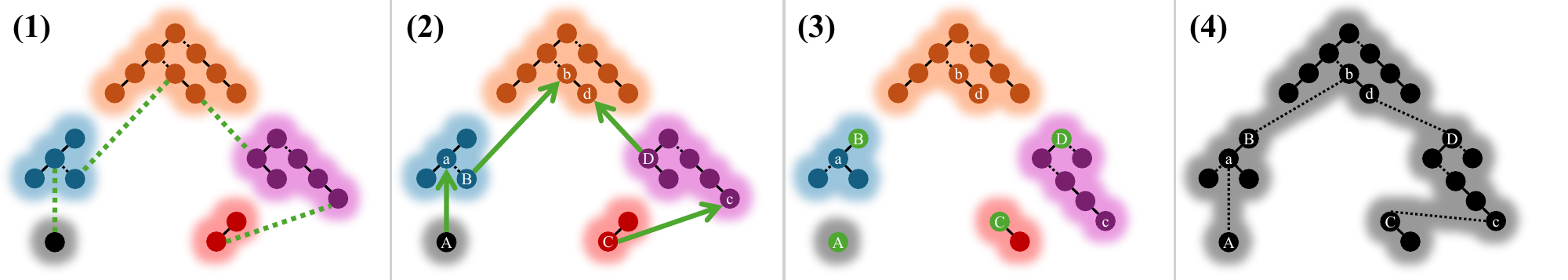}
    \caption{\small An illustration of the batch link algorithm for link-cut trees. (1) shows the existing link-cut tree components in various colors, with green edges representing new links. (2) shows the directed links after the algorithm roots the link tree. In (3), the source vertex of each directed link (upper case letters) is made the root of its component with $\evert$. In (4) the non-preferred parent for each source vertex is set to the destination vertex (lower case letters), linking all the components. Then the algorithm restores Invariant~\ref{inv:heavy} (not depicted).}
    \label{fig:batch_link}
\end{figure*}

\subsection{Efficient Batch-Parallel LCT Updates}
We now present batch-parallel update algorithms that achieve optimal efficiency regardless of input topology. When implemented with batch-parallel biased zipzip trees, these algorithms satisfy Theorem~\ref{thm:lct_analysis} (and also maintain highly efficient bounds parameterized by the input diameter $D$, as proved in \appref{Appendix~\ref{app:lct_cost}}).

\begin{restatable}{theorem}{lctbatchupdate}\label{thm:lct_analysis}
    In a link-cut tree with the heavy preferred child invariant implemented with biased zipzip trees, the batch link and batch cut algorithms take $O(k \log (1+n/k))$ expected work and $O(\log n)$ depth with high probability where $k$ is the batch size. Queries take $O(\log n)$ time with high probability.
\end{restatable}

Our algorithms maintain the same invariant as the sequential link-cut tree version with worst-case update cost guarantees~\cite{sleator1983data}. Let $n(v)$ denote the \defn{size} (number of vertices) of the subtree rooted at $v$ in the link-cut tree. A vertex $v$ is defined as a \defn{heavy child} of its parent $p$ if $2 \cdot n(v) > n(p)$. This guarantees that each vertex has at most one heavy child.
The invariant is as follows:

\begin{invariant}[Heavy Preferred Child] \label{inv:heavy}
    A vertex $v$ with parent $p$ is the preferred child of $p$ if and only if $v$ is a heavy child of $p$.
\end{invariant}

\begin{lemma} \label{lem:few_np}
    In a link-cut tree with the heavy preferred child invariant, the number of non-preferred edges on the path from any vertex to the root is at most $\lceil \log_2 n \rceil$.
\end{lemma}

\myparagraph{Weight Augmentation}
Updates often alter vertex sizes, requiring preferred child changes to maintain Invariant~\ref{inv:heavy}. Explicitly tracking $n(v)$ is inefficient since $\reverse$ operations can alter the size of each vertex along a path. Instead, we augment the auxiliary trees so each element $x$ tracks a \defn{weight}, $w(x)$, defined as one plus the total size of its \emph{non-preferred} children. Consequently, the size $n(v)$ is the sum of weights from $v$ to the tail of its preferred path. By augmenting the auxiliary tree with subtree weight sums, we can compute this suffix weight efficiently (detailed in \appref{Appendix~\ref{app:treap_query}}).
\revision{In \appref{Appendix~\ref{app:lct_query}}, we show that using this weight as the biased weight of elements in our biased zipzip trees, makes the total length of any path to the root $O(\log n)$ \whp{}.
}

\myparagraph{Child Sets}
When a preferred child is cut or its size decreases, we must find the vertex's new heavy child. Storing a vertex's non-preferred children in a standard ordered set imposes an $\Omega(\log n)$ update overhead, destroying the ability to bound the runtime efficiently in terms of the input diameter $D$. To bypass this, we introduce a \defn{heavy element set} (detailed in \appref{Appendix~\ref{app:lct_child_set}}).
It supports $O(1)$ expected-time updates, $O(1)$ worst-case queries to report an element exceeding half the total set sum, and batch updates of size $k$ in $O(k)$ expected work and $O(\log k)$ depth.
The value of each child $c$ in a child set is $n(c)$.
Combining this set with our weight augmentation yields the following primitive:
\begin{itemize}[leftmargin=*,topsep=0pt,itemsep=0pt]
    \item $\fname{GetHeavyChild}(v)$: Returns the heavy child of $v$, or $\vname{null}$. It computes the current preferred child's size via suffix sum, temporarily inserts it into $v$'s heavy element set, queries the set for the true heavy child, and reverts the insertion.
\end{itemize}

\myparagraph{Efficient Batch Link}
Algorithm~\ref{alg:lct_batch_link} shows our theoretically efficient batch link algorithm. Similar to the simple algorithm, we first direct the links by computing a rooted forest over the link trees (line~\ref{line:root_link_trees}). For each directed link $(u, v)$, we independently call $\evert(u)$ and set $u.\vname{nppar} = v$ (lines~\ref{line:evert_link_start}--\ref{line:evert_link_end}). Assuming the heavy child invariant held initially, these independent sequential everts correctly maintain the invariant within the existing components.
The remainder of the algorithm restores the invariant globally.

\begin{algorithm}[t]
\small
\caption{$\fname{BatchLink}(\vname{links})$}
\label{alg:lct_batch_link}
    $E \gets \fname{RootLinkTrees}(\vname{links})$ \label{line:root_link_trees}\;
    \parfor{$(u, v) \in E$ \label{line:evert_link_start}} {
        $\fname{Evert}(u)$ \;
        $u.\vname{nppar} = v$ \label{line:evert_link_end}\;
    }
    Update weights of vertices. \label{line:update_weights}\;
    \parfor{$[v, \vname{group}] \in \fname{GroupByDst}(E)$ \label{line:insert_children_start}} {
        $v.\fname{InsertChildren}(\vname{group})$ \label{line:insert_children_end}
    }
    $\vname{splits} \gets \{\}$, $\vname{joins} \gets \{\}$ \label{line:pseudocode_duplicate_start} \tcp*{parallel bags}
    $W \gets $ set of vertices with a reweighted child. \label{line:reweighted_child_start}\;
    \parfor{$v \in W$} {
        \If{$\fname{GetHeavyChild}(v) \neq \getsucc(v)$}{
            Add $(v, \getsucc(v))$ to $\vname{splits}$. \;
            Add $(v, \fname{GetHeavyChild}(v))$ to $\vname{joins}$. \label{line:reweighted_child_end}
        }
    }
    $\batchsplit(\vname{splits})$ \label{line:lct_batch_split}\;
    $\batchjoin(\vname{joins})$ \label{line:lct_batch_join} \label{line:pseudocode_duplicate_end}
\end{algorithm}

Next, we update vertex weights (line~\ref{line:update_weights}). Weights increase in two ways: \defn{directly} (receiving new non-preferred children via the new links) or \defn{indirectly} (when a descendant in a non-preferred subtree experiences a direct increase). 
We compute direct increases using the Euler tour technique~\cite{tarjan1985efficient} on the rooted link trees to find the sum of component sizes added to each vertex. For indirect increases, we propagate these direct weight changes upward along non-preferred edges to the component root (prior to the new links). We model this propagation as a directed tree where edges $(x,y)$ connect a vertex $x$ to its preferred path head's non-preferred parent $y$. A second Euler tour~\cite{tarjan1985efficient} over this tree computes the subtree sums, yielding the total indirect weight updates. Finally, all weight changes are applied via a batch augmented value update on the sequence.

We then group the directed edges by destination using parallel semisort~\cite{gu2015top} and batch-insert the new children into their respective heavy element child sets (lines~\ref{line:insert_children_start}--\ref{line:insert_children_end}).
The value for each new child set entry $x$ is its new size $n(x)$ computed in the previous step.
The values of existing children in other child sets must also be updated to reflect any indirect weight changes.

With weights and child sets updated, we identify vertices requiring preferred child changes. We prove in \appref{Appendix~\ref{app:lct_correct}} that in a batch link, only reweighted vertices can violate the invariant. For each reweighted vertex $v$ (lines~\ref{line:reweighted_child_start}--\ref{line:reweighted_child_end}), we check if its current preferred child ($\getsucc(v)$) matches its heavy child ($\fname{GetHeavyChild}(v)$). If they differ, we schedule a split to remove the old child and a join to attach the new one. Finally, we execute a $\batchsplit$ to sever the incorrect preferred children (line~\ref{line:lct_batch_split}) and a $\batchjoin$ to connect the correct ones (line~\ref{line:lct_batch_join}). We also update the weights of affected parent vertices accordingly.

\myparagraph{Efficient Batch Cut}
We provide the pseudo-code for our batch cut algorithm in \appref{Appendix~\ref{app:lct_details}}.
The additions over the simple batch cut algorithm are similar to the batch link algorithm. We determine the vertices that require weight updates, update child sets, then check which vertices need to change their preferred child, then change those preferred children with a $\batchcut$ followed by a $\batchlink$.
Unlike $\batchlink$, there may be some vertices violating the invariant that were not themselves reweighted (rather a vertex below them in their preferred path lost weight, decreasing the size of the current preferred child).
In \appref{Appendix~\ref{app:lct_details}} we provide more details for finding such vertices, and we prove the correctness and efficiency of the batch cut algorithm.

\section{Experimental Evaluation}

All experiments were run on a 96-core machine (192 hyper-threads) with 4 $\times$ 2.1 GHz Intel Xeon(R) Platinum 8160 CPUs (each with 33MiB L3 cache) and 1.5TB of main memory.
Our implementations are written in C++ and compiled with -O3 optimization.
Our parallel implementations use ParlayLib~\cite{parlaylib}, a C++ library for shared memory parallel programming. We use mimalloc~\cite{leijen2019mimalloc} and the allocator from ParlayLib for parallel memory allocation.
We provide details about our implementations in \appref{Appendix~\ref{app:experiments}}.

\subsection{Dynamic Sequence Performance} \label{sec:seq_exp}
For dynamic sequences we compare our new batch-parallel treap data structure (MOJOS-Treap) against existing implementations of batch-parallel skip list~\cite{tseng2019batch} (PAR-SL) and an existing top-down batch-parallel treap implementation~\cite{tseng2019batch} (PAR-Treap) which is not work-efficient for batch split.
Each of these implementations also provides sequential join and split operations which we use for our first experiment.
We do not compare against join-based tree frameworks like PAM~\cite{pam} since they do not implement batch join and batch split as required for our framework.
Parallel $(a,b)$-trees~\cite{akhremtsev2016fast} do fit our batch-dynamic sequence interface, however there is no publicly available implementation to our knowledge.

\myparagraph{Sequential Update and Query Speed}
We test the performance of sequential updates by generating a random permutation of $n=10^8$ elements to define a global sequence order. First we join each subsequent element in the order, doing each join in a random order. Then we split each element from its successor in a random order. We report the total time for each operation type.
For query speed we pick $q=10^8$ random elements and measure the total time it takes to travel to the root from each element. All queries occur between the join and split phases.
Table~\ref{tab:sequence_seq} shows the results.

\begin{table}[ht]
\small
    \centering
    \begin{tabular}{|l|l l|l l|}
        \hline
        Data Structure & Joins & Splits & Updates & Queries \\
        \hline
        MOJOS Treap & 10.822 & 14.041 & 24.863 & 233.753 \\
        Top-Down Treap~\cite{tseng2019batch} & 30.725 & 55.076 & 85.800 & 254.162 \\
        Skip List~\cite{tseng2019batch} & 47.436 & 67.019 & 114.454 & 360.198 \\
        \hline
    \end{tabular}
    \caption{\small The total time in seconds for sequential dynamic sequence operations in various data structures.}
    \label{tab:sequence_seq}
\end{table}

Our bottom-up treap algorithms are significantly faster for both joins and splits than the top-down treap algorithms. In this setting, the top-down treap algorithms are less efficient because they require traversing to the root of the joined or split elements first.
The skip list performs worse than both treap algorithms, likely because its structure results in more nodes traversed on average on the path from an element to its root.
This hypothesis is further supported by the worse query performance of skip lists, since query operations directly traverse to the root with no write operations.

\myparagraph{Batch-Parallel Update Speed}
For batch updates, we partition the joins and splits into $n/k$ batches of size roughly $k$, with varying values of $k$.
Figure~\ref{fig:sequence_batch_update} shows the results of these experiments.

\begin{figure}
    \centering
    \includegraphics[width=0.9\linewidth]{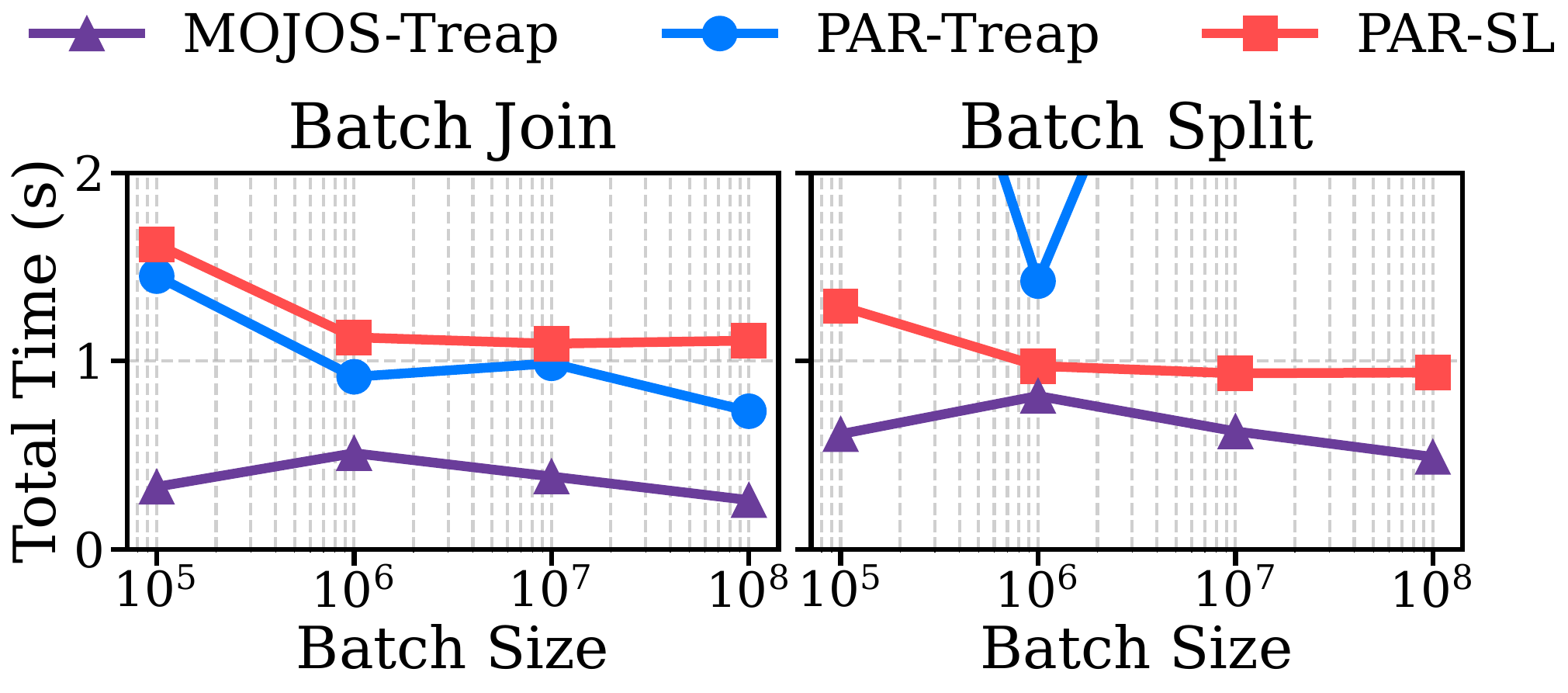}
    \caption{\small The times of various data structures on batch joins (left) and batch splits (right) using various batch sizes on a sequence with $n=10^8$. Each data point represents the total time to do $n$ total joins followed by $n$ total splits, partitioned into batches of size roughly $k$.}
    \label{fig:sequence_batch_update}
\end{figure}

For both batch joins and batch splits, MOJOS-Treap requires less total time than the other baselines across all input configurations we tested.
These results indicate that it is possible for the superior performance of bottom-up treap algorithms in the sequential setting to also translate into the batch-parallel setting.

Compared to PAR-SL, MOJOS-Treap is $3.371\times$ faster on average (geometric mean) for batch joins and $1.639\times$ faster on average (geometric mean) for batch splits.
The reason this difference is smaller for batch split is that the MOJOS-Treap algorithm runs in two concurrent phases requiring additional synchronization overhead, whereas PAR-SL supports phase-concurrent splits.

The top-down PAR-Treap algorithm has reasonable performance for batch joins, outperforming PAR-SL but not our bottom-up MOJOS-Treap.
These results highlight the benefits of bottom-up treap algorithms in this setting: less node traversal and lower synchronization costs.
The PAR-Treap batch split algorithm performs significantly worse than all others due to work-inefficiency.

\myparagraph{Memory Usage}
A treap on $n$ elements contains exactly $n$ nodes. In our MOJOS-Treap implementation, each node stores only a parent pointer, two child pointers, and a 64 bit priority field, totaling to 32 bytes per node. Thus the memory usage is $32n$ bytes plus a small constant factor.
The PAR-Treap implementation stores some extra fields per node used during their updating algorithms, our implementation avoids these overheads.
The skip list has in expectation $2n-1$ nodes for a sequence of $n$ elements.
For our experiments with $n=10^8$, MOJOS-Treap used approximately 3.2GB of memory, and PAR-Treap and PAR-SL both used approximately 4.8GB of memory.

\begin{figure*}
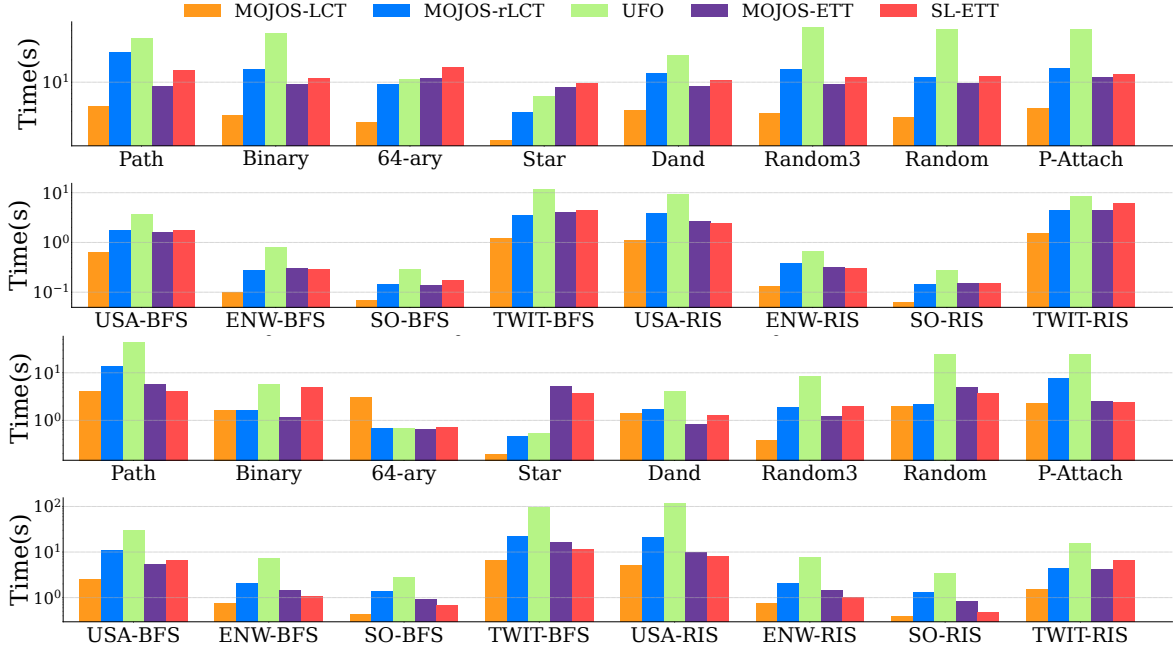

    \centering
    \includegraphics[width=0.9\linewidth]{figures/pdt_update_speed.pdf}
    \includegraphics[width=0.9\linewidth]{figures/pdt_rotate_update_speed.pdf}
    \caption{\small
    The results of our batch-dynamic tree update speed experiments. We use a fixed batch size of $k=10^6$ for all experiments.
    The first row shows results on synthetic trees with $n=10^8$ using the build-destroy update pattern.
    The second row shows results on breadth-first search forests and random incremental spanning forests of our real-world graph datasets using the build-destroy update pattern.
    The third row shows results on synthetic trees with $n=10^7$ using the separate-reconnect update pattern.
    The fourth row shows the results on the real-world spanning forests using the separate-reconnect update pattern.
    }
    \label{fig:pdt_batch_update}
\end{figure*}

\begin{figure}
    \centering
    \includegraphics[width=0.9\linewidth]{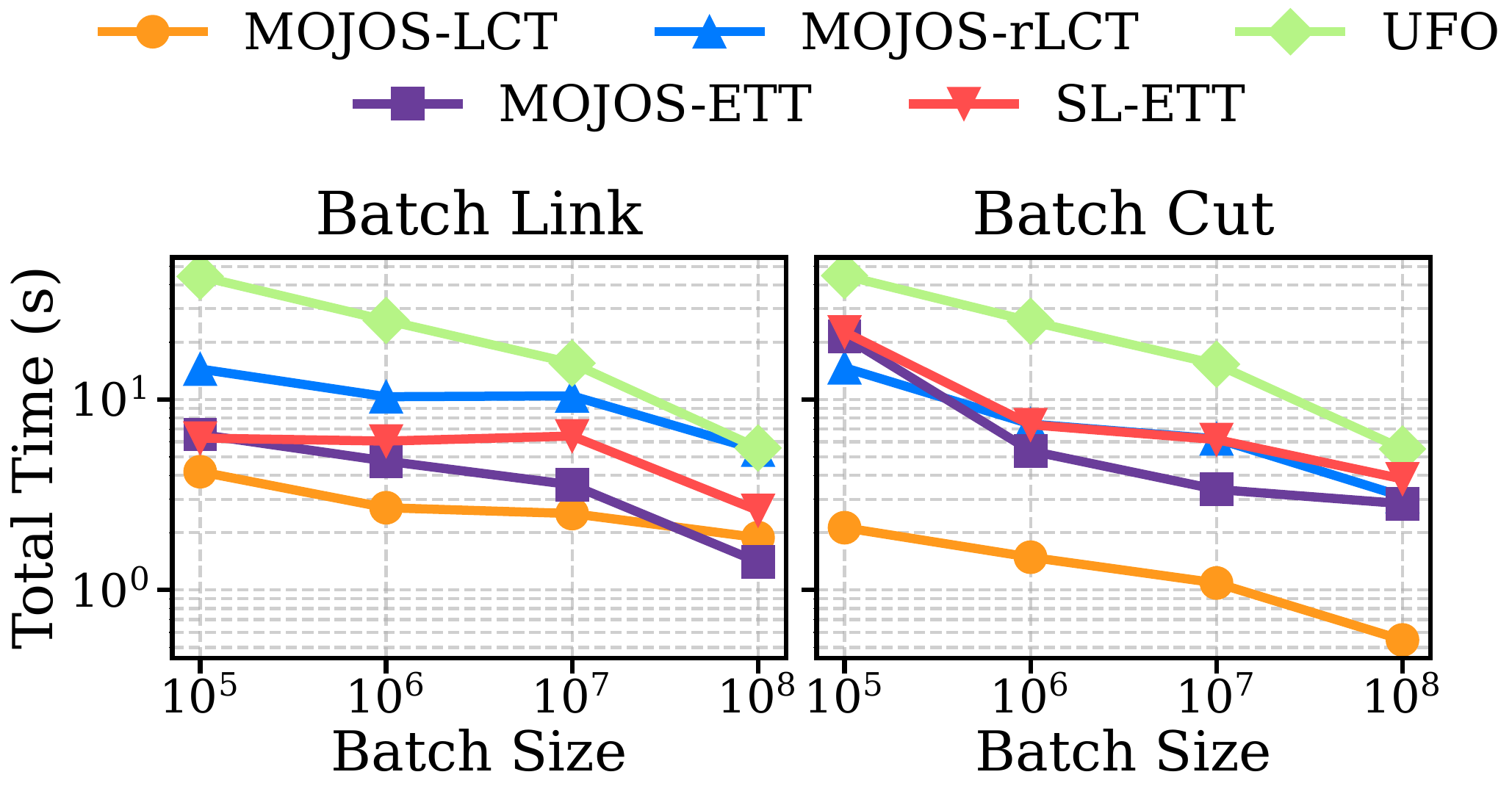}
    \caption{\small The times of dynamic trees on batch links (left) and batch cuts (right) with various batch sizes on a random tree with $n=10^8$ using the build-destroy pattern. Each point represents the total time to do $n-1$ links or cuts, partitioned into batches of size roughly $k$.}
    \label{fig:pdt_batch_size_sweep}
\end{figure}

\subsection{Dynamic Trees Performance}
For dynamic trees we evaluate three of our new implementations: the batch-parallel Euler tour tree algorithm optimized for treaps (MOJOS-ETT), the simple batch-parallel link-cut tree algorithm (MOJOS-LCT), and the batch-parallel link-cut tree algorithm {\em robust} to adversarial cases (MOJOS-rLCT).
We compare these against existing implementations of batch-parallel Euler tour trees based on skip lists~\cite{tseng2019batch} (SL-ETT), and UFO trees~\cite{deman2026ufo} (UFO).
We omit other implementations of topology trees~\cite{deman2026ufo} and rake-compress trees~\cite{ikram2025parallel} since they were shown to be significantly slower than UFO trees and ETTs across a variety of inputs~\cite{deman2026ufo}.

To compare performance on a wide range of trees with different attributes, we test our implementations on several types of synthetic input trees, including path graphs, perfect binary trees, perfect k-ary trees, stars, and dandelions.
We also test on randomly generated degree 3, unbounded degree, and preferential attachment trees.
We also test our implementations on spanning trees of several real-world graphs.
Due to space constraints, we describe our graph inputs in \appref{Appendix~\ref{app:experiments}}.
For each graph, we test both a {\em breadth-first spanning forest (BFS)} starting at a random vertex, and a {\em random incremental spanning forest (RIS)} generated by inserting the edges of the graph in a random order and only including edges that are not connected in the current spanning forest.

\myparagraph{Sequential Update Speed}
Prior work has extensively studied the performance of sequential dynamic tree data structures~\cite{deman2026ufo}.
Their results showed that link-cut trees are consistently the fastest data structure.
The fast implementation of link-cut trees that is commonly used is based on splay trees~\cite{sleator1985self} which are purely sequential.
Since our batch-parallel link-cut trees instead use a treap as the underlying data structure, we provide some experimental results to show that a sequential treap-based link-cut tree implementation can be as fast as the standard splay tree implementation.
Across all of our synthetic input cases with $n=10^8$, our treap-based LCT implementation was $1.011\times$ faster on average (geometric mean) than the splay tree LCT implementation.

\myparagraph{Batch-Parallel Update Speed}
\revision{
We measure the performance of batch-parallel updates by studying two different update patterns.
The first update pattern is \defn{build-destroy}.
This first links all of the edges in a random order, and then cuts them all in a random order.
The links and cuts are partitioned into batches of size roughly $k$.
The second update pattern is \defn{separate-reconnect}.
This starts with the complete tree, and cuts a random batch of roughly $k$ edges, then links that same batch. This is repeated for disjoint batches of roughly $k$ edges, until all edges have been cut and linked once.
}
Figure~\ref{fig:pdt_batch_update} shows the results of these experiments.

The first three data structures in our evaluation are batch-dynamic trees supporting path queries (MOJOS-LCT, MOJOS-rLCT, and UFO trees).
Across all inputs we see that both of our LCT implementations process batch updates significantly faster than UFO trees which are the current state-of-the-art. MOJOS-LCT performs very well on these inputs, but as discussed in the next section, there exist adversarial cases on which MOJOS-LCT performs poorly.
When robustness to input is crucial, MOJOS-rLCT is the fastest available option for path queries.
If it is known that the inputs are low diameter or have non-adversarial structure, MOJOS-LCT is preferred.

The last three data structures in our evaluation are batch-dynamic trees supporting subtree queries (UFO trees, MOJOS-ETT, and SL-ETT).
Here we see that the ETT implementations generally outperform UFO trees, except on low diameter inputs.
For all but a few inputs, MOJOS-ETT outperforms SL-ETT.

\revision{Figure~\ref{fig:pdt_batch_size_sweep} shows that the relative performance of these data structures stays similar across various batch sizes, and that all data structures perform better with higher batch sizes.}

\begin{figure}
    \centering
    \includegraphics[width=0.9\linewidth]{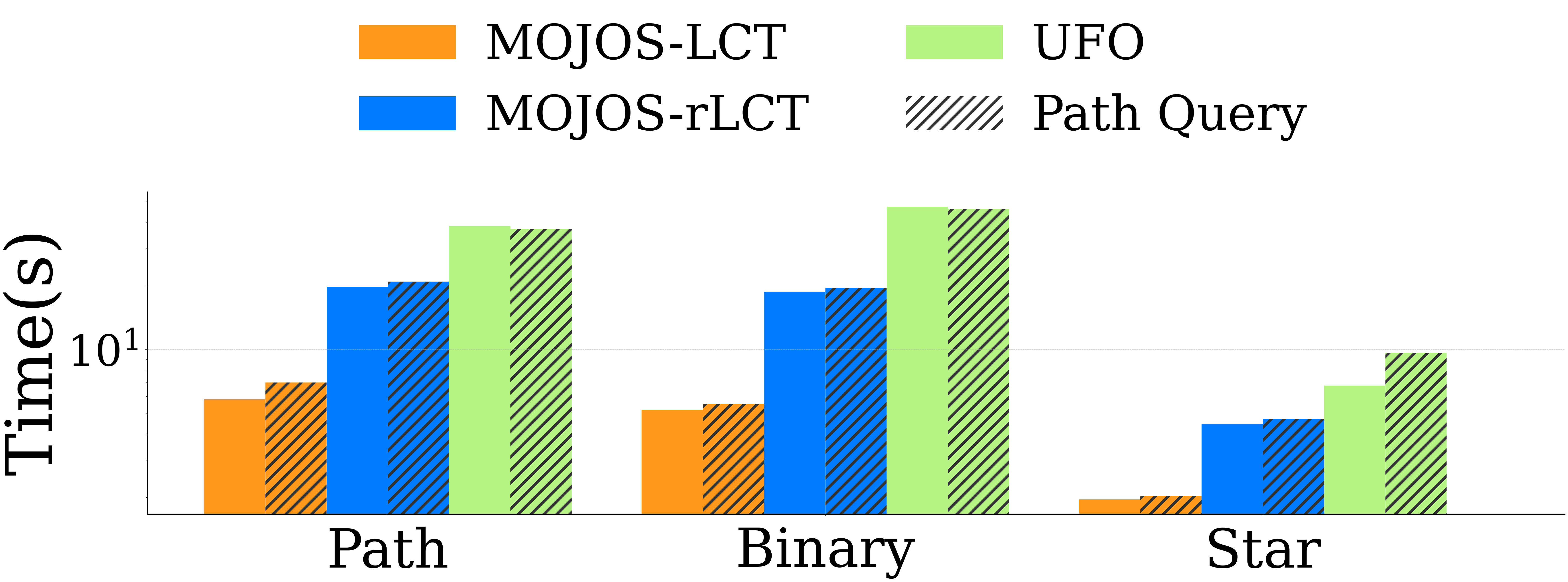}
    \caption{\small A subset of our batch-dynamic tree update speed experimental results for data structures with (solid bar) and without (hashed bar) path query support. We use trees with $n=10^8$ and a batch size of $k=10^6$.}
    \label{fig:path_query_update_speed}
\end{figure}

\begin{figure}
    \centering
    \includegraphics[width=0.9\linewidth]{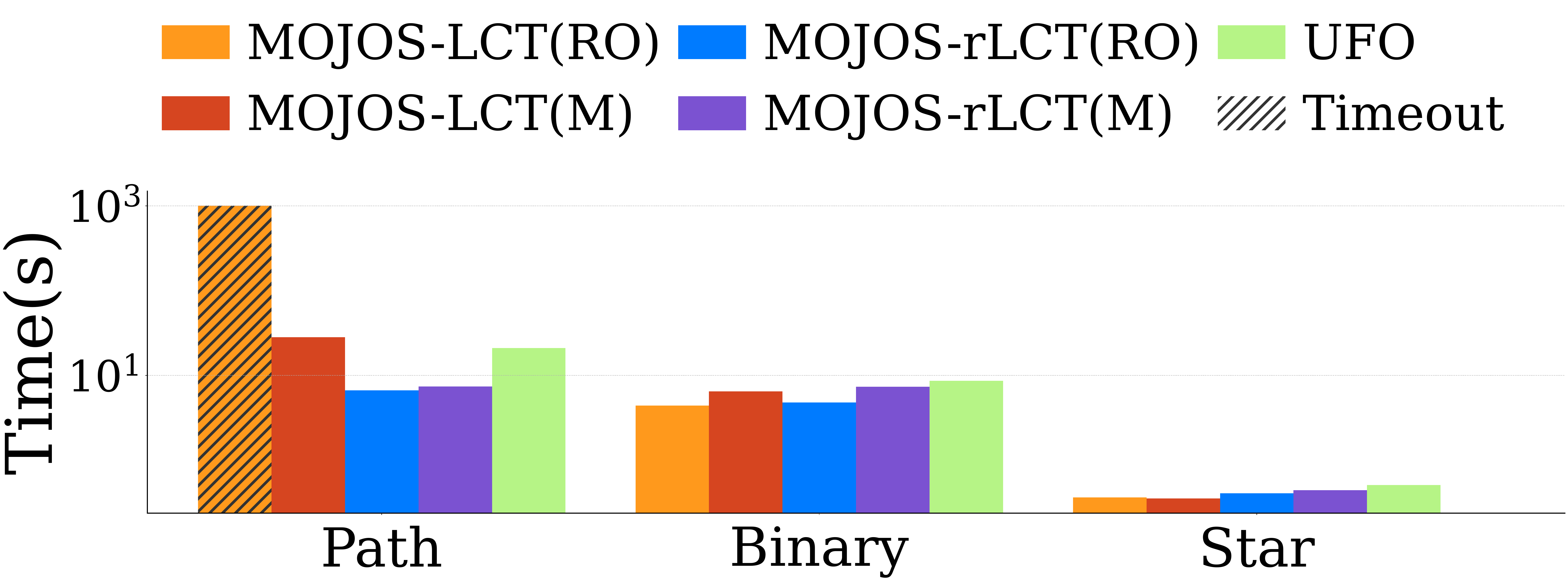}
    \caption{\small The results of our sequential path query speed experiments. We use trees with $n=10^8$, first building the entire tree using a batch link operation, then reporting the time to execute $q=10^6$ uniformly randomly generated path queries. ``(RO)'' indicates our read-only LCT query algorithms.
    ``(M)'' indicates the standard mutating algorithm. The hashed bar indicates a timeout after 1000 seconds.}
    \label{fig:sequential_path_query}
\end{figure}

\begin{figure}
    \centering
    \includegraphics[width=0.8\linewidth]{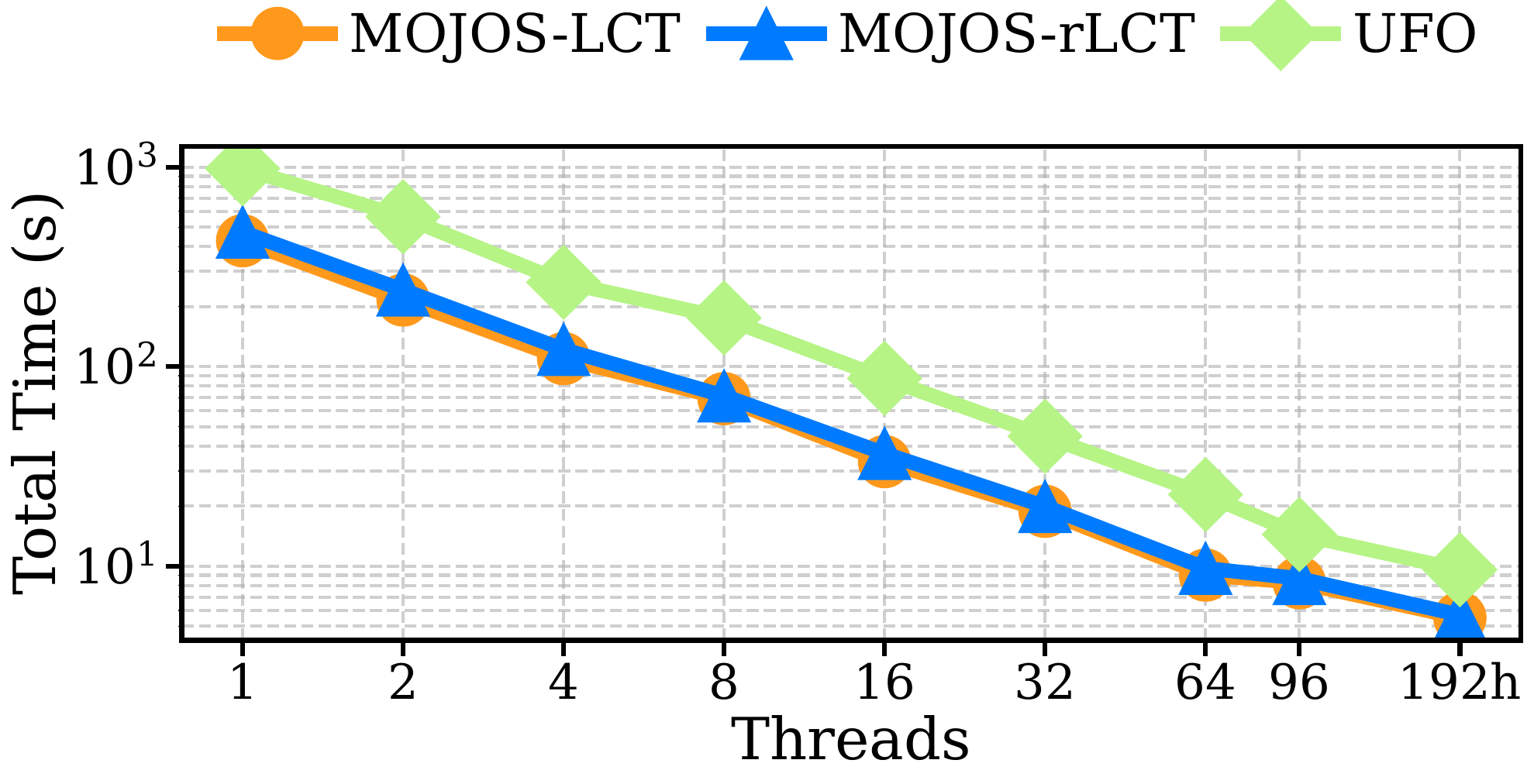}
    \caption{\small The results of our parallel path query speed experiments. We use a random tree with $n=10^8$, first building the entire tree using a batch link operation, then reporting the time to execute $q=10^8$ uniformly randomly generated path queries. All queries are run in one parallel batch using the read-only query algorithms for each data structure. The $x$-axis varies the number of threads available.}
    \label{fig:parallel_path_query}
\end{figure}

\revision{
\myparagraph{Path Query Performance}
We study various aspects of the performance of batch-dynamic tree data structures supporting path sum queries with 32-bit integer weights per vertex.

Figure~\ref{fig:path_query_update_speed} studies the affect of supporting path queries on the update speed of the data structures.
The results show that supporting typically results in a small overhead in update speed.
However, the relative update speeds between the three data structures remains the same with and without path query support.

Figure~\ref{fig:sequential_path_query} studies the sequential speed of various path query algorithms.
For MOJOS-LCT and MOJOS-rLCT, we test both the traditional mutating link-cut tree path query algorithm, as well as the read-only algorithm that we proposed.
The results show that the read-only algorithms are slightly faster, except in the path case for MOJOS-LCT. MOJOS-LCT is not theoretically efficient in this case, so the mutating algorithm with efficient amortized bounds is significantly better. We discuss these adversarial cases below.
UFO tree path queries are slightly slower in non-adversarial cases.

Finally, Figure~\ref{fig:parallel_path_query} shows the scalability of our parallel query algorithms (the read-only algorithms in a parallel for loop) with various thread counts. All 3 data structures achieve excellent speedups with increasing numbers of threads available.

}

\myparagraph{Adversarial Case for MOJOS-LCT}
Our MOJOS-LCT implementation is based on the simple batch-parallel LCT algorithm with no strong theoretical guarantee for high diameter input trees. Thus it is possible to construct adversarial cases for which the performance of MOJOS-LCT suffers. In contrast, MOJOS-rLCT is based on our algorithm with good theoretical guarantees regardless of input, and it does not lose performance on such cases.

The adversarial case that we evaluate is a path graph, where the first operation is a batch link that adds all of the edges in a long path, and then repeated random query operations occur.
For MOJOS-LCT, this input results in all edges being non-preferred, thus queries must traverse through the entire path one node at a time.
On this input with $n=10^5$ and doing $q=10^5$ random queries, MOJOS-rLCT was $49.341\times$ faster than MOJOS-LCT.
Cases like this highlight the importance of our robust LCT algorithm which maintains the heavy preferred child invariant.

\myparagraph{Memory Usage}
On a tree with $n=10^8$, MOJOS-LCT uses 4.0GB, MOJOS-rLCT uses 4.8GB, MOJOS-ETT uses 12.5GB, and SL-ETT uses 16.5GB.
The LCT implementations have much lower memory usage since they only store exactly $n$ treap nodes, while ETT implementations store up to $3n-2$ elements.
MOJOS-rLCT uses slightly more memory than MOJOS-LCT due to storing augmented weight information for vertices.
Of the ETT implementations, SL-ETT uses significantly more memory since the skip-list stores $6n-5$ nodes in expectation while the treap stores exactly $3n-2$ nodes.

\section{Conclusion}

In this paper, we developed MOJOS, a framework for building fast and theoretically-efficient batch-dynamic trees from batch-dynamic sequences.
A natural next step is to perform a comprehensive empirical evaluation of how different sequence implementations (including those yet to be developed) impact MOJOS's performance.
Additionally, we demonstrated that batch-parallel link-cut trees significantly outperform other batch-dynamic trees. While the simple algorithm is highly efficient, it remains sensitive to adversarial inputs. Conversely, the robust algorithm incurs high overhead. Investigating lightweight heuristics to improve the simple algorithm's robustness without the performance penalty of strictly maintaining theoretical invariants is a promising direction for future study.

\clearpage
\bibliographystyle{ACM-Reference-Format}
\bibliography{main.bib}

\clearpage

\appendix
\section{Treap Algorithms and Proofs}

\subsection{Bottom-Up Sequential Treap Algorithms} \label{app:treap_seq}
Here we describe sequential algorithms for treap join and split which employ a \defn{bottom-up} approach rather than the iterative top-down approach commonly seen in the literature on treaps~\cite{seidel1996randomized, blelloch2016just}.
These bottom-up algorithms form the basis of the phase-concurrent batch join and batch split algorithms in this paper. They additionally have some practical benefits over prior algorithms.
First, since \myjoin and \mysplit take sequence elements as arguments, the bottom-up approach requires only one traversal of a path in the treap, whereas top-down algorithms require first traversing up to the root and then performing the top-down algorithm.
Second, our bottom-up algorithms perform a number of writes proportional to the number of nodes whose parent changes. The top-down algorithms perform a number of writes proportional to the length of the path, regardless of whether many pointers actually need to change.
Finally, our algorithms are simple and direct, avoiding any recursion or iteration overhead of a top-down approach.
The pseudo-code for our algorithms are shown in Algorithms~\ref{alg:bottomupjoin}~and~\ref{alg:bottomupsplit}.
The algorithms take time proportional to the height of the treap which is $O(\log n)$.

\begin{algorithm}[ht]
\caption{$\mathsf{Join}(x,y)$}
\label{alg:bottomupjoin}
$X = x$, $Y = y$\;
\While{$X$ \textbf{and} $Y$}{
    \If{$X.p < Y.p$}{
        \While{$X.par$ \textbf{and} $X.par.p < Y.p$}{
            $X = X.par$\;
        }
        $next_X = X.par$\;
        $X.par = Y$\;
        $Y.lchild = X$\;
        $X = next_X$\;
    }
    \ElseIf{$X.p > Y.p$}{
        \While{$Y.par$ \textbf{and} $Y.par.p < X.p$}{
            $Y = Y.par$\;
        }
        $next_Y = Y.par$\;
        $Y.par = X$\;
        $X.rchild = Y$\;
        $Y = next_Y$\;
    }
}
\end{algorithm}


\myparagraph{Sequential Join}
Algorithm~\ref{alg:bottomupjoin} shows the algorithm for treap join. The algorithm receives $x$, the rightmost node in one treap, and $y$, the leftmost node in an other treap. The algorithm essentially ``merges'' the spine of $x$ with the spine of $y$ in increasing order of priority.
We define the \defn{spine} of a node in a treap to be the path from that node to the root. Since $x$ and $y$ are the rightmost and leftmost node in their treap respectively, their spines are the same as the right spine and left spine respectively of their treap.
The merge starts with the lower priority node in either spine, and advances up the spine to the first ancestor whose parent has larger priority than the current node in the other spine (or the root if there is none). This node's parent is set to the current node in the other spine and becomes its left or right child. Then we repeat this process on the opposite spine. We keep alternating between the spines in this way until the root of one is reached. 
This effectively merges the two spines into one, such that the priority of the nodes in the new spine is in increasing order, maintaining the max-heap property.
Whenever a node in the spine of $x$ changes its parent to a node in the spine of $y$, it becomes the left child of its new parent. Similarly nodes in the spine of $y$ become the right child of nodes in the spine of $x$. This ensures that that the inorder sequence of nodes in the new treap is precisely the inorder sequence of the nodes in $x$'s treap followed by those in $y$'s treap.

\begin{algorithm}[ht]
\caption{$\mathsf{Split}(x,y)$}
\label{alg:bottomupsplit}
$X = x$, $Y = y$\;
\If{$X.p < Y.p$}{
    $Y = y$, $X = y.lchild$\;
    $X.par = null$, $Y.lchild = null$\;
}
\ElseIf{$X.p > Y.p$}{
    $X = x$, $Y = x.rchild$\;
    $Y.par = null$, $X.rchild = null$\;
}
\While{$X$ \textbf{and} $Y$}{
    \If{$X.p < Y.p$}{
        \If{$Y.par$ \textbf{and} $Y.par.rchild == Y$}{
            $X.par = Y.par$, $Y.par.rchild = X$\;
            $X = Y.par$, $Y.par = null$\;
        }
        \Else{
            $Y = Y.par$\;
        }
    }
    \ElseIf{$X.p > Y.p$}{
        \If{$X.par$ \textbf{and} $X.par.lchild == X$}{
            $Y.par = X.par$, $X.par.lchild = Y$\;
            $Y = X.par$, $X.par = null$\;
        }
        \Else{
            $X = X.par$\;
        }
    }
}
\end{algorithm}

\myparagraph{Sequential Split}
Algorithm~\ref{alg:bottomupsplit} shows the algorithm for treap split. Since $x$ is directly before $y$ in the sequence order, either $x$ is an ancestor of $y$ in the treap or vice versa (otherwise their LCA would be between them in the order). We can determine which one is higher by their priorities, the larger priority node being an ancestor of the smaller priority node.
For the description of this algorithm, assume without loss of generality that $x$ was above $y$. The right child of $x$ is the root of the subtree containing $y$, call this $Y$. The algorithm will proceed up the spine of $Y$, ``unmerging'' the spine based on which side of the split each node is on.
Let $X$ be $x$. Then $X$ and $Y$ will be the bases of the two spines formed from the unmerging.
As we proceed up the spine, we keep track of the direction we are traversing (variable $dir$). If a node is its parent's left child, we call this traversing right ($dir = 0$). Otherwise we are traversing left ($dir = 1$).
Since $Y$ is a right child, we begin traversing left.
As long as the direction stays the same while advancing up the spine, all nodes traversed are on the same side of the split, and no pointers need to change.
Once we change direction while advancing up the spine, the next parent node should be on the other side of the split. The last node encountered on this other side should become a child of this new parent. This process proceeds until we reach the root.
Since the order of nodes on the two new spines respects the order of the original spine, the max-heap property is maintained.
To keep the correct sequence order, nodes in the new right spine must become left children of their new parent and vice versa.

\subsection{Batch-Parallel Treap Cost Bound Proofs} \label{app:treap_cost}

\batchtreaptheorem*
\begin{proof}
    We begin by proving that the depth of batch join is $O(\log n)$. Forking into $k$ threads takes $O(\log k)$ depth which is $O(\log n)$ since $k$ must be at most $n-1$.
    The depth of each thread is at most proportional to the deepest path in any treap after the join plus the number of failed CAS operations. The length of any treap path is $O(\log n)$.
    The CAS can only fail on the root of a treap which is involved in two joins. If the CAS fails, the other join succeeded, thus the next CAS for this join will succeed. Thus the number of failed CAS operations is at most one per treap node, $O(\log n)$. Since each of the $k$ threads has depth $O(\log n)$ the overall depth of the algorithm is $O(\log n)$.

    Next we will prove that the work of batch join is $O(k \log (1+n/k))$. \revision{It is well-known that the number of distinct nodes in $O(k)$ element to root paths in a set of treaps with $n$ total elements is $O(k \log (1+n/k))$ in expectation~\cite{Blelloch1998}.} We will show that the work done is proportional to the number of nodes in $O(k)$ element to root paths in a set of treaps.
    These $O(k)$ paths are all the spines of the joined elements in the treaps before the batch of joins. Consider two joins around a shared center treap. Let the left treap be $A$, the center treap be $B$, and the right treap be $C$. These two joins operate on separate sets of elements until they reach the root of the center treap, $R$.
    One of the joins (without loss of generality, the left join) will end with a successful CAS operation on the parent pointer of $R$. At this point the left join will have only done work proportional to the number of nodes in $B$'s left spine and the number of nodes in $A$'s right spine whose priority is less than $R$'s priority. 
    The remaining join (without loss of generality, the right join) will do work proportional to the number of nodes in $B$'s right spine, the the number of nodes in $C$'s left spine, and the number of nodes in $A$'s right spine whose priority is greater than $R$'s priority.
    This argument extends naturally to $k$ joins and a constant amount of work can be charged to each node in the spines of all the joined elements.
\end{proof}

\revision{
\begin{lemma} \label{lem:batch_split_helper}
Let $\mathcal{T}$ be a complete binary tree with $L = \Theta(k)$ leaves. For any leaf $u$, let the \textit{monotone spine} $S(u)$ be the length of the maximal path of edges ascending from $u$ that does not change direction. 
Let $W_u$ be a non-negative random variable assigned to each leaf $u$ such that $\sum_{u \in \mathcal{T}} E[W_u] = k$. If the expected weight of every leaf is bounded by $E[W_u] \le C = O(n/k)$ (where $n \ge k$), then the expected total cost satisfies:
\[
E\left[ \sum_{u \in \mathcal{T}} W_u S(u) \right] = O\left(k \log \left(1 + \frac{n}{k}\right)\right)
\]
\end{lemma}

\begin{proof}
By the linearity of expectation, we can move the expectation operator inside the summation:
\[
E\left[ \sum_{u \in \mathcal{T}} W_u S(u) \right] = \sum_{u \in \mathcal{T}} E[W_u] S(u)
\]
For shorthand, let $w_u = E[W_u]$ be the deterministic expected weight of leaf $u$. We are given that $\sum_{u \in \mathcal{T}} w_u = k$ and $w_u \le C = O(n/k)$.

The height of $\mathcal{T}$ is $H = \Theta(\log k)$. In a complete binary tree, the number of leaves with a monotone spine of length at least $\ell$ is bounded by $N(\ell) = O(k / 2^\ell)$, because specifying a monotone spine of length $\ell$ fixes the bottom $\ell$ bits of the leaf's path index. Consequently, the number of leaves with a spine of length \textit{exactly} $\ell$ is also bounded by $O(k / 2^\ell)$.

To find the worst-case configuration of expected weights that maximizes the total sum subject to the capacity constraint $C$, we employ a greedy strategy: we assume that the leaves with the largest possible spine lengths are assigned their maximum possible expected weight $C$. 

Let $\ell^*$ be the threshold spine length such that the total capacity of all leaves with spine length $\ge \ell^*$ is at least $k$. Because $n \ge k$, the number of distinct leaves we need to saturate is $m \approx k / C = \Theta(k^2 / n)$. We determine $\ell^*$ by setting the number of available leaves at or above this threshold equal to $m$:
\[
N(\ell^*) \approx m \implies \frac{k}{2^{\ell^*}} \approx \frac{k^2}{n} \implies 2^{\ell^*} \approx \frac{n}{k}
\]
Solving for the threshold yields $\ell^* \approx \log(n/k)$.

We compute the upper bound on the total expected cost by summing across all active lengths from $\ell^*$ up to the tree height $H$:
\[
\sum_{u \in \mathcal{T}} w_u S(u) \le \sum_{\ell = \ell^*}^{H} \left( \text{number of leaves with length } \ell \right) \cdot C \cdot \ell
\]
Substituting the bounds for the number of leaves and the capacity $C = O(n/k)$:
\[
\sum_{u \in \mathcal{T}} w_u S(u) \le \sum_{\ell = \ell^*}^{H} O\left(\frac{k}{2^\ell}\right) \cdot O\left(\frac{n}{k}\right) \cdot \ell = O(n) \sum_{\ell = \ell^*}^{H} \frac{\ell}{2^\ell}
\]
Using the standard arithmetico-geometric series identity $\sum_{\ell = \ell^*}^{\infty} \frac{\ell}{2^\ell} = O\left(\frac{\ell^*}{2^{\ell^*}}\right)$, we evaluate the bound:
\[
O(n) \cdot O\left(\frac{\ell^*}{2^{\ell^*}}\right) = O(n) \cdot O\left( \frac{\log(n/k)}{n/k} \right) = O\left(k \log \left(\frac{n}{k}\right)\right)
\]
Adding a constant to the logarithm to gracefully handle the boundary case where $n \approx k$ yields the final clean bound:
\[
E\left[ \sum_{u \in \mathcal{T}} W_u S(u) \right] = O\left(k \log \left(1 + \frac{n}{k}\right)\right)
\]
\end{proof}
}

\batchtreaptheorem*
\begin{proof}
    Here we analyze the depth of batch split.
    The depth of the first phase is determined by the maximum depth/cost in a call to $\fname{ConcurrentSplit}$.
    The maximum number of steps in this function is proportional to the number of nodes traversed upwards, plus the number of times a CAS is retried.
    It is well known that the number of nodes traversed upwards in a path in a treap is bounded by $O(\log n)$ with high probability.
    If a CAS is ever retried, another thread marked the child pointer first, so the thread will terminate after either a successful CAS, or the current child has a lower priority already. Therefore the number of total retries is bounded by the total number of retries possible at a single node. This is $O(\log n)$ with high probability because all the nodes attempting to CAS the same child pointer must lie on a single spine in the treap.
    The depth in the second phase is logarithmic in the number of attempted child pointer changes, \revision{which is at most $n$}, so the depth is $O(\log n)$.
    Thus the overall depth is $O(\log n)$ with high probability.

    Now we analyze the work of batch split.
    The work of the first phase is the sum of the number of nodes traversed plus the number of CAS retries for each thread. It is clear that the work in the second phase is asymptotically no more than that in the first phase.
    
    First we consider the total cost in nodes traversed. We must account for the fact that multiple threads can have overlapping traversals through paths that are not changing direction. However once the path changes direction, only one thread will proceed because of the child pointer marking logic.
    \revision{In a treap, the expected size of any subtree rooted at depth $H = \lfloor \log_2 k \rfloor$ is $O(n / 2^H) = O(n/k)$~\cite{seidel1996randomized}. The expected length of a search path within a randomized treap of size $M$ is known to be $O(\log M)$. Therefore, the expected number of nodes traversed by any single thread below depth $H$ is bounded by $O(\log(1 + n/k))$.}
    In the worst-case, each split traverses completely independent paths below this depth \revision{$H$, for a total of $O(k \log (1+n/k))$ nodes traversed in expectation.}

    Now consider the remaining \revision{nodes traversed at depths $\leq H$.} It is easy to see that the total number of nodes is $O(k)$, but we need to bound the \revision{total number of nodes in the partially overlapping paths traversed in our algorithm.}
    \revision{Let $W_u$ be a non-negative random variable representing the number of thread paths passing through a boundary leaf $u$ at depth $H$. Because there are $k$ threads and each passes through exactly one leaf at depth $H$, $\sum E[W_u] = k$. Due to treap properties, the expected number of keys in any subtree rooted at depth $H$ yields $E[W_u] = O(n/k)$. The total overlapping traversal cost in this top zone is given by $\sum W_u S(u)$. Applying Lemma~\ref{lem:batch_split_helper} directly, the expected number of unique nodes traversed across all threads in this zone is $O(k \log (1+n/k))$.}

    Finally, we will show that the total number of CAS retries \revision{in the entire algorithm} is upper bounded by the total number of nodes traversed, \revision{implying that CAS retries do not asymptotically increase the total work.}
    Consider a node $P$ on which $x$ of its descendants are trying to CAS its child pointer. At least one of the descendants will succeed the CAS in every round, thus the maximum number of CAS operations on this pointer is $\sum_{i=1}^x i$. Prior to attempting this CAS, each descendant must have traversed up from a different node on the spine of the subtree rooted at the previous child of $P$. The total number of nodes traversed is minimized when all the nodes are packed at the top of the spine, but then there are still at least $\sum_{i=1}^x i$ total nodes traversed.

    Putting this all together, the total number of nodes traversed, and the total number of CAS operations are both $O(k \log (1+n/k))$ \revision{in expectation}, so the work bound holds.
\end{proof}

\subsection{Augmented Values in Batch-Parallel Treaps} \label{app:augmented_treap}
Augmented values associate a data value from a domain $D$ with each element and maintain the aggregate of the entire subtree at each node using an associative and commutative function $f: D \times D \to D$, enabling efficient operations such as range queries. Similar techniques to those for batch-parallel skip lists~\cite{tseng2019batch} can be used to maintain these augmented values in batch-parallel treaps. This adds an additional parallel phase to both the batch join and batch split algorithms, but does not impact the asymptotic cost bounds (assuming the aggregate function $f$ can be computed in $O(1)$ time).

The augmented value recomputation phase proceeds in two steps to safely parallelize the updates. First, originating from the up to $k$ nodes modified by the batch operation, we perform a concurrent bottom-up traversal. Threads traverse parent pointers up to the root, marking each ancestor node on the spine. To prevent race conditions from concurrent threads attempting to mark the same ancestor, this flag is set using a Compare-And-Swap (CAS) operation. 

Once the affected spine is marked, the algorithm initiates a recursive top-down traversal starting from the root. This traversal only explores children that possess the mark. The actual recomputation occurs recursively in a post-order fashion: once a marked node's children have finished recomputing their own values, the current node recomputes its aggregate using $f$ and clears its mark. Because the set of marked nodes is exactly the union of at most $k$ paths to the root, we know that this two-step phase requires only $O(k \log(1+n/k))$ \revision{expected} work and $O(\log n)$ depth with high probability~\cite{Blelloch1998}, ensuring the overall asymptotic cost of the batch operations remains unaffected.

Range queries in the batch-parallel treap can be answered efficiently without mutating the underlying tree structure. To process a single range query for the subsequence between elements $i$ and $j$, the algorithm identifies the lowest common ancestor (LCA) of these elements. From this LCA, the algorithm traverses the boundary paths down to $i$ and $j$, accumulating the augmented values of the maximal disjoint subtrees that exactly cover the query range, alongside the values of the individual nodes on those paths. Because the treap has $O(\log n)$ depth with high probability, this read-only traversal correctly evaluates the aggregate in $O(\log n)$ time.

\subsection{MOJOS Queries in Treaps} \label{app:treap_query}

Here we describe the following batch query operations for treaps. They also extend directly to zipzip trees and biased zipzip trees.

\begin{itemize}[leftmargin=15pt]
    \item $\batchgetrep(Q)$
    \item $\batchgetpred(Q) / \batchgetsucc(Q)$
    \item $\batchgethead(Q) / \batchgettail(Q)$
\end{itemize}

\myparagraph{Batch Representative Queries}
The $\batchgetrep$ algorithm is similar to the algorithm used for skip lists~\cite{tseng2019batch} and other data structures.
First spawn a parallel thread for each element in the batch. That thread traverses up from the element, marking a boolean field in each node with an atomic CAS operation to set the mark as true. If the CAS fails, the thread stops. If the CAS succeeds it proceeds to the next parent in the treap.
Then for each root of the treaps containing these elements, proceed top-down. For each marked child, recursively in parallel, remove the mark and continue to your children. Once you reach a node for an element in the query input batch, write the name of the original root.

The work of this algorithm is proportional to the number of nodes on the union of $k$ paths to the root in a collection of treaps, and the depth is proportional to the maximum depth of an element.
Applying known treap properties, the work is $O(k \log (1+n/k))$ \revision{in expectation} and the depth is $O(\log n)$ with high probability.

\myparagraph{Batch Predecessor and Successor Queries}
The sequential algorithm for $\getsucc$ of an element $x$ in a treap is as follows. If $x$ has a right child, return the head of the subtree rooted at the right child of $x$ (by repeatedly following left child pointers).
If $x$ has no right child, we must return the first ancestor of $x$ that is to the right. This is done by traversing up through parent pointers until a node is found where the previous node is its left child. If this never occurs, there is no successor and the algorithm returns $\vname{null}$.
The $\getpred$ algorithm is defined symmetrically.

Since the sequential algorithm is read-only, $\batchgetsucc$ and $\batchgetpred$ can be implemented correctly by running the sequential algorithm independently in parallel for each element in the input batch.
The following lemma proves that this method actually takes $O(k \log (1+n/k))$ \revision{expected} work. The depth is $O(\log n)$ with high probability.

\begin{lemma}
    $\batchgetpred$ and $\batchgetsucc$ operations in a treap take $O(k \log (1+n/k))$ \revision{expected} work and $O(\log n)$ depth with high probability.
\end{lemma}
\begin{proof}
    The total work is the sum of the nodes traversed by the $k$ independent queries. Each $\getsucc$ (or $\getpred$) query for an element $x$ strictly traverses the tree path from $x$ to its successor (or predecessor). Because the key intervals between distinct queried elements and their immediate successors are non-overlapping, these tree paths are entirely edge-disjoint. Thus, the total work is proportional to the total number of unique nodes in the union of these paths.
    This set of nodes is a subset of the union of all paths from the $k$ queried elements and their corresponding successors (at most $2k$ distinct nodes) to the root. The union of these $2k$ paths contains $O(2k \log(1 + n/(2k))) = O(k \log (1+n/k))$ nodes \revision{in expectation~\cite{Blelloch1998}}, which bounds the total work.
    Since the algorithms run independently in parallel, the span (depth) is simply bounded by the maximum length of a single query path. It is well-known that the maximum depth of a treap is $O(\log n)$ with high probability, guaranteeing an $O(\log n)$ depth for the batch operations.
\end{proof}

It is also possible to directly store a predecessor and successor pointer for each element in the treap. If this is the case, predecessor and successor queries take $O(1)$ time, and batch queries take $O(k)$ work and $O(\log k)$ depth.
Maintaining this pointers only requires the join, split, batch join, and batch split functions to update these pointers correspondingly when two successive elements are joined or split apart.

\myparagraph{Batch Head and Tail Queries}
Batch head and tail queries can be done using a batch representative query, and then for each representative finding the head or tail of its treap. By a similar analysis as before, this takes $O(k \log (1+n/k))$ \revision{expected} work and $O(\log n)$ depth with high probability.

Additionally, it is possible to augment treaps such that each node stores the head and tail of the subtree rooted at it in the treap. As described previously, this augmentation can be maintained efficiently during updates to the treap. This means that given the root (representative) of a treap, the head and tail can be found in $O(1)$ time. This will be a useful property in our later analysis of our link-cut tree algorithms.

\subsection{Batch-Parallel Biased Zipzip Trees} \label{app:zipzip}
Biased zipzip trees maintain a weight $w(x)$ for each element $x$. They provide the property that the expected depth of an element $x$ is $O(\log(W/w(x)))$ where $W$ is the total weight.

All of our algorithms for batch-parallel treaps can be used directly to implement batch-parallel biased zipzip trees. The only difference is in how the priorities between two elements are compared. Our algorithms can easily be changed to compare priorities to respect the rules of biased zipzip trees.
The work of batch join, batch split, and batch representative queries is proportional to the number of nodes on the union of the $O(k)$ paths to the root. If the total weight, $W$, in the biased zipzip tree is $W=O(n)$, this work is $O(k \log (1+n/k))$ in expectation. The depth of each algorithm is proportional to the maximum length of one of the paths.
If $W=O(n)$, the depth is $O(\log n)$ \whp{}.

Since biased zipzip trees maintain weights for each element, it is natural to provide the ability to change the weight of an element. We define the following batch operation for updating weights:
\begin{itemize}[leftmargin=15pt]
    \item $\batchupdate(U)$: given a list of elements in a biased zipzip trees, and new weights for each element, update the weight of each element to its new value.
\end{itemize}
Since the structure of the biased zipzip tree is dependent on the weights of the elements, this operation requires some structural modifications. We implement the operation as follows.
First determine the predecessor and successor for each element. Use a $\batchsplit$ operation to split each element from its predecessor and successor. This results in each of the updated elements being in their own biased zipzip tree. Now update the weight field in each node in parallel.
Then use a $\batchjoin$ operation to undo the splits done in the first part.
The $\batchjoin$ on the reweighted elements ensures that the structure of the new biased zipzip tree respects the new weights.
By similar arguments as the proofs for the cost of $\batchjoin$ and $\batchsplit$ in treaps, the work of this operation is proportional to the number of nodes on the union of $k$ paths, and the depth is proportional to the maximum depth of one of these elements in the biased zipzip tree.

\section{Parallel Ordered Set} \label{app:set}

\subsection{Parallel Ordered Set in MOJOS}

\myparagraph{Implementing Batch Find}
Here we describe how treaps (with $n$ elements) implement the $\batchfind$ operation on $k$ values in $O(k \log (1+n/k))$ \revision{expected} work and $O(\log n)$ depth \whp{}.

\begin{itemize}[leftmargin=15pt]
    \item $\find(x) / \batchfind(Q)$: Given a sequence representative, returns the element with the maximum key less than or equal to $x$. The batch variant processes a pre-sorted list of queries.
\end{itemize}

A single $\find(x)$ operation starts from the root as $\vname{curr}$. If the value in node $\vname{curr}$ equals $x$, return $\vname{curr}$.
Else if the target value $x$ is less than the value in node $\vname{curr}$, proceed recursively into the left subtree setting $\vname{curr} = \vname{curr}.lchild$. If no such child exists return the last visited node for which we proceeded into the right subtree (or $\vname{null}$ if no such node exists).
Otherwise $x$ is greater than the value at $\vname{curr}$, so proceed recursively into the right subtree.
This takes $O(\log n)$ time \whp{}, because of the maximum depth of the treap.

The typical algorithm for $\batchfind$ on a pre-sorted list of values works as follows.
At each node, use binary search on the node's value to partition the query list into two lists with values less than or greater than this node (for queries with value equal, return this node).
Then recurse into the left child and right child with the lesser list and the greater list respectively.
This is proven to take $O(k \log (1+n/k))$ \revision{expected} work and $O(\log n \log k)$ depth \whp{}. The depth incurs an additional $\log k$ factor for binary search at every level of the treap.

Our improved algorithm only takes $O(\log n)$ depth \whp{}.
First we traverse down $\log k$ levels at a time, extracting the values of every node up to depth $\log k$ into a sorted list. This list has size at most $O(k)$, so the work and depth of this step are $O(k)$ and $O(\log k)$ respectively.
Then we use a parallel merge~\cite{jaja1992parallel} to merge the this list of treap values with our query batch (for equal values treat treap values as lower than query values).
Next, we can use a parallel scan~\cite{blelloch1990pre} to determine the index of the next treap value after each treap value in the merged list.
The range of values from each treap value in the merged list up until but not including the next treap value, represents a sublist of query values that specifically will fall into one of the subtrees rooted at a node with depth $\log k + 1$. This subtree is rooted at either the right child of the node with treap value directly before this sublist in the merged order, or the left child of the node with treap value directly to the right of this sublist (whichever one of these actually has depth $\log k + 1$).
Each sublist proceeds to search in their corresponding subtree in the second phase (or values that were equal to the treap value proceeding the sublist return the node with that treap value).
Up to this point all operations have taken $O(k)$ work and $O(\log k)$ depth.

\revision{Then we simply run the same algorithm recursively in parallel for each subtree and the range of values to be found in that subtree (note that the value $k'$ in the recursive sub-problem is the new number of elements in this subtree, not the original batch size).}

\revision{
To analyze the complexity, observe that in any recursive phase with $k'$ queries, the algorithm extracts and processes up to $\log k'$ levels of the current subtree using $O(k')$ work and $O(\log k')$ depth. Because $O(k')$ work is distributed across $k'$ queries traversing up to $\log k'$ levels, the cost amortizes to $O(1)$ work per query per tree level descended. Consequently, the total work across all recursive steps is asymptotically proportional to the sum of the individual search path lengths for all $k$ queries. Since the expected sum of search path lengths for $k$ keys in a randomly structured treap of size $n$ is bounded by $O(k \log (1+n/k))$, the total expected work of our improved algorithm matches this $O(k \log (1+n/k))$ bound. Furthermore, because each phase requires $O(\log k')$ parallel depth to descend $O(\log k')$ tree levels, the parallel depth incurred along any search path is directly proportional to the length of that path in the tree. As the maximum depth of a treap with $n$ elements is $O(\log n)$ \whp{}, the total parallel depth of the algorithm is strictly bounded by $O(\log n)$ \whp{}.
}

\myparagraph{Optimal Batch Insertion and Batch Deletion in MOJOS}
The algorithm for batch insertion first calls $\batchfind$ on the list of values to insert. It filters out values that are already in the treap.
Then, for every node returned from $\batchfind$ its successor is determined using $\batchgetsucc$.

Then $\batchsplit$ is used to separate each of the found nodes from its successor.
Next, treap nodes are initialized for all of the new elements to insert.
Finally, $\batchjoin$ is used to join each new element into the treap. Specifically, each for each element $x$, two joins are created (and duplicates are removed).
The first join is either $(\find(x), x)$ or if the value $a$ preceding $x$ in the join batch is greater than $\find(x)$, the join is $(a,x)$.
The second join is either $(x, \getsucc(\find(x))$ or if the value $b$ succeeding $x$ in the join batch is less than $\getsucc(\find(x))$, the join is $(x,b)$.

It is easy to see that with the costs of the MOJOS operations for batch-parallel treaps, and the cost of $\batchfind$ described earlier, these simple algorithms achieve both optimal work and depth: $O(k \log (1+n/k))$ \revision{expected} work and $O(\log n)$ depth \whp{}.

The algorithm for batch deletion is similar. First call $\batchfind$ to find the nodes in the treap corresponding to each value in the input batch, discarding values that do not appear in the treap.
Next use $\batchsplit$ to disconnect the predecessor and successor of each node that will be deleted.
Nodes for deleted values are freed.

To reconnect the split treaps, we need to join the node to the left of every range of deleted nodes to the element on the right. These joins can be computed with list contraction, and executed with $\batchjoin$.
Once again it is easy to see that these simple algorithms achieve both optimal work and depth: $O(k \log (1+n/k))$ \revision{expected} work and $O(\log n)$ depth \whp{}.

\subsection{Ordered Set Experimental Results}

Here we experimentally evaluate our implementation of the parallel ordered set algorithms from the previous section (MOJOS-Set). We compare against parallel augmented maps (PAM)~\cite{pam}, the state-of-the-art for parallel ordered sets.
Our experiments indicate that (despite little optimization) our simple approach in MOJOS-Set achieves reasonable performance compared to the highly optimized PAM implementation.
The results are in Table~\ref{tab:ordered_set_experiments}.
MOJOS-Set is $3.906\times$ slower on average (geometric mean) than PAM for batch insertions, and $3.257\times$ slower on average (geometric mean) for batch deletions.

\begin{table}[ht]
    \centering
    \begin{tabular}{l l l l l}
        \toprule
        \multicolumn{5}{c}{Batch Insertion} \\
        \midrule
        Batch Size & $10^6$ & $5 \cdot 10^6$ & $10^7$ & $5 \cdot 10^7$ \\
        \midrule
        MOJOS-Set & 2.857 & 1.330 & 0.813 & 0.485 \\
        PAM & 0.801 & 0.340 & 0.257 & 0.092 \\
        \bottomrule
    \end{tabular}
    \begin{tabular}{l l l l l}
        \toprule
        \multicolumn{5}{c}{Batch Deletion} \\
        \midrule
        Batch Size & $10^6$ & $5 \cdot 10^6$ & $10^7$ & $5 \cdot 10^7$ \\
        \midrule
        MOJOS-Set & 3.486 & 1.726 & 1.184 & 0.537 \\
        PAM & 1.353 & 0.696 & 0.457 & 0.079 \\
        \bottomrule
    \end{tabular}
    \caption{The total time in seconds for batch insertions and batch deletions with varying batch size and $n=10^8$ for MOJOS-Set compared to PAM.}
    \label{tab:ordered_set_experiments}
\end{table}
\section{Batch-Parallel Euler Tour Trees}

\subsection{MOJOS ETT Batch Cut} \label{app:ett_cut}

The batch cut algorithm first concurrently splits before and after the Euler tour tree element for each cut edge in both directions. In MOJOS, each component that will receive at least one split, must join its tail to its head after the splits happen (unless there will be split just to the right of the tail). This assures that the Euler tour sequences after this phase of splits are exactly the same as in the original algorithm.
Algorithmically, this is accomplished by calling $\getrep$ for each split, semisorting it to group by component, and then for each component store its tail and head before doing the splits (also checking that there is no split call on the tail element already).

The next part of the algorithm again joins the Euler tours fragments together into cyclic tours using concurrent join calls. In MOJOS, we use the exact same method as for batch link to omit a single join per component, avoiding cycles.

\subsection{MOJOS Batch-Parallel ETT Analysis} \label{app:ett_analysis}

\ettanalysis*
\begin{proof}
    We analyze the work and depth of the batch link algorithm step-by-step. Let $k = |\vname{links}|$ be the batch size, and $n$ be the total number of vertices in the Euler tour tree.
    We begin by analyzing the initialization phase where new directed edge elements are instantiated for the incoming links. Creating these $2k$ elements and inserting them into the hash table~\cite{gil1991towards} takes $O(k)$ expected work and $O(\log k)$ depth \whp. Subsequently, the algorithm identifies the necessary split points by querying the underlying sequence data structure for the successor, tail, and head of the affected vertices. This results in at most $2k$ queries.
    By our assumption on the batch-dynamic sequence, processing this batch of $O(k)$ queries takes $O(k \log(1+n/k))$ \revision{expected} work and $O(\log n)$ depth \whp{}.
    To group the newly created directed edges by their source vertex, the algorithm employs a parallel semisort. This sorting step requires $O(k)$ expected work and $O(\log k)$ depth \whp{}.
    
    Once the edges are grouped, the algorithm iterates over them in parallel to determine the necessary joins between the sequence fragments. Performing $O(1)$ work per iteration across $2k$ iterations takes $O(k)$ work and $O(\log k)$ depth.
    Executing the gathered splits involves calling the batch split operation on the underlying sequence. Since there are at most $O(k)$ splits collected, this operation takes $O(k \log(1+n/k))$ \revision{expected} work and $O(\log n)$ depth \whp{}.
    
    Before connecting the fragments, the algorithm must remove potential cycles among the joins. This requires finding the representative for each of the $O(k)$ join endpoints. Mapping endpoints to their representatives is another batch of $O(k)$ sequence queries, taking $O(k \log(1+n/k))$ \revision{expected} work and $O(\log n)$ depth \whp{}. Identifying and omitting the cycle-inducing joins is done via parallel list contraction on a graph of $O(k)$ components, which takes $O(k)$ work and $O(\log k)$ depth \whp{}. Finally, executing the batch join on the remaining $O(k)$ valid joins takes $O(k \log(1+n/k))$ \revision{expected} work and $O(\log n)$ depth \whp{}.
    
    Summing the costs of these phases, the expected work is dominated by the batch sequence operations and the randomized semisort, yielding $O(k \log(1+n/k))$ expected work. The depth is bounded by the sequence operations, resulting in $O(\log n)$ depth \whp{}.

    We analyze the work and depth of the batch cut algorithm by tracing its execution. Let $k$ be the number of edges to be cut, and $n$ be the total number of vertices in the Euler tour tree.
    First, the algorithm identifies the necessary split points to isolate the edges being removed. For each of the $k$ cut edges in both directions, it concurrently identifies splits before and after the corresponding Euler tour tree element, taking $O(k)$ work and $O(\log k)$ depth. To determine which components are affected and ensure they maintain a valid Euler tour structure, the algorithm calls $\getrep$ for each split point. Processing these $O(k)$ queries to the underlying sequence data structure takes $O(k \log(1+n/k))$ \revision{expected} work and $O(\log n)$ depth \whp{}.
    
    To group the splits by their respective components, the algorithm uses a parallel semisort based on the component representatives. This sorting step requires $O(k)$ expected work and $O(\log k)$ depth \whp{}.
    Once grouped, the algorithm processes each affected component in parallel to identify its tail and head elements for the corrective join, while verifying that the tail element is not already scheduled for a split. Finding these boundaries requires another batch of $O(k)$ sequence queries, taking $O(k \log(1+n/k))$ \revision{expected} work and $O(\log n)$ depth \whp{}.
    With the splits and the tail-head joins fully identified, the algorithm executes the batch split operation on the underlying sequence to isolate the cut edges. For the $O(k)$ collected splits, this operation takes $O(k \log(1+n/k))$ \revision{expected} work and $O(\log n)$ depth \whp{}.
    
    Finally, the algorithm reconnects the remaining fragments into valid cyclic tours using concurrent joins. Just as in the batch link procedure, it must avoid creating cycles by omitting exactly one join per component. This cycle removal requires mapping elements to their representatives via $O(k)$ $\getrep$ queries, taking $O(k \log(1+n/k))$ \revision{expected} work and $O(\log n)$ depth \whp{}, followed by a parallel list contraction on $O(k)$ components, taking $O(k)$ work and $O(\log k)$ depth \whp{}. Executing the final batch join operation on the filtered list of $O(k)$ joins takes $O(k \log(1+n/k))$ \revision{expected} work and $O(\log n)$ depth \whp{}.
    
    Summing the costs of these phases, the expected work is dominated by the batch sequence operations and the randomized semisort, yielding $O(k \log(1+n/k))$ expected work. The depth is bounded by the sequence operations, resulting in $O(\log n)$ depth \whp{}.
\end{proof}

\subsection{ETT Augmented Values and Subtree Queries} \label{app:augmented_ett}
Augmented Euler tour trees associate a data value from a domain $D$ with each vertex in the tree. Given an associative and commutative function $f: D \times D \to D$, augmented Euler tour trees can answer subtree queries which return the result of $f$ applied over the values associated with all vertices in a particular subtree.

Our MOJOS ETT algorithms also support augmented values and subtree queries assuming that the batch-dynamic sequence data structure used supports augmented values and range queries (see Appendix~\ref{app:augmented_treap} for details). The algorithms are the same as \cite{tseng2019batch} which just reduces subtree queries to  range queries on the sequence.

\subsection{Cyclic Treap for Batch-Parallel ETT} \label{app:cyclic_treap}

Algorithm~\ref{alg:cyclicconcurrentjoin} shows the modified concurrent join algorithm that supports inputs that cause cycles in the sequences represented by the treaps. The changes from the original concurrent join algorithm (Algorithm~\ref{alg:concurrentjoin}) are highlighted in green.

\begin{algorithm}[ht]
\caption{$\fname{CyclicConcurrentJoin}(x,y)$}
\label{alg:cyclicconcurrentjoin}
$X = x$, $Y = y$\;
\While{$X$ \textbf{and} $Y$ \colorbox{mygreen}{\textbf{and} $X \neq Y$}}{
    \If{$X.prio < Y.prio$}{
        \While{$X.par$ \textbf{and} $X.par.prio < Y.prio$}{
            $X = X.par$\;
        }
        $next_X = X.par$\;
        $Y.lchild = X$\;
        \If{$\mathsf{CAS}(\&X.par, next_X, Y)$}{
            $X = next_X$\;
        }
    }
    \ElseIf{$X.prio > Y.prio$}{
        \While{$Y.par$ \textbf{and} $Y.par.prio < X.prio$}{
            $Y = Y.par$\;
        }
        $next_Y = Y.par$\;
        $X.rchild = Y$\;
        \If{$\mathsf{CAS}(\&Y.par, next_Y, X)$}{
            $Y = next_Y$\;
        }
    }
}
\tikzmarkin[fill=mygreen, draw=mygreen]{highlight1}(5,-0.4)(0,0.4)
\If{$X = Y$}{
    \If {$X.lchild$ \textbf{and} $X.lchild$ \text{not linked from the left}}{
        $X.lchild = \vname{null}$
    }
    \If {$X.rchild$ \textbf{and} $X.rchild$ \text{not linked from the right}}{
        $X.rchild = \vname{null}$
        \tikzmarkend{highlight1}
    }
}
\end{algorithm}

In a standard batch-parallel setting, cycle-inducing joins either cause infinite loops during root traversals or fundamentally corrupt the tree structure. This is why the MOJOS algorithm requires preprocessing steps to filter out cycle-forming joins. Algorithm~\ref{alg:cyclicconcurrentjoin} bypasses this expensive preprocessing by automatically detecting and resolving cycles on the fly during the concurrent join phase.

The core modification lies in the early termination condition $X \neq Y$ and the subsequent cleanup block (highlighted in green). When a batch of joins inadvertently creates a cycle (for instance, linking the tail of an Euler Tour sequence back to its head), the concurrent left and right traversal pointers, $X$ and $Y$, will traverse the sequence and eventually collide at the exact same node.

Upon detecting $X = Y$, the algorithm immediately halts the traversal, preventing an infinite loop. However, because this is a concurrent environment, threads may have partially installed child pointers via $\mathsf{CAS}$ before realizing the structure had looped. The cleanup block resolves this by inspecting the child pointers of the collision node $X$. It verifies the integrity of $X.lchild$ and $X.rchild$ by checking if they are properly linked from the expected logical directions. If a child pointer does not have a corresponding valid sequence relationship (i.e., it represents the erroneous edge that closes the cycle), the algorithm safely nullifies it. This natively guarantees that exactly one join operation in the cycle-inducing component is discarded, leaving a perfectly valid, acyclic Euler Tour sequence without needing prior preprocessing.

\subsection{High Priority ETT Vertex Elements} \label{app:high_prio_vertex_ett}

Batch-parallel Euler tour trees have elements for both vertices and the directed edges between them. For a tree containing $n$ vertices, the ETT must maintain $n$ vertex elements and up to $2n-2$ edge elements, totaling up to $3n-2$ elements in the underlying treap. In a standard augmented treap, every single element must allocate memory to store the augmented aggregate of its respective subtree. This can greatly increase the memory usage, especially if the augmented values are particularly large.
Using many augmented values or large augmented values is common in use cases such as dynamic connectivity~\cite{holm2001poly,acar2019parallel}, and sketch-based dynamic connectivity~\cite{kapron2013dynamic,gibb2015dynamic} where each vertex has a large sketch object as augmented value.

We exploit the max-heap property of the treap to mathematically eliminate this overhead. By setting the most significant bit of the randomized priority field to $1$ for all vertex elements and $0$ for all edge elements, we enforce the strict invariant that $prio(v) > prio(e)$ for any vertex $v$ and edge $e$.
Because the treap structurally enforces that parent priorities must be greater than child priorities, this invariant forces all $n$ vertex elements to percolate upward. They form a contiguous cluster at the top of the tree that guarantees the root is always a vertex element. Conversely, all $2n-2$ edge elements appear in the lower branches of the treap.

The critical consequence of this invariant is that no edge element can ever contain a vertex element in its subtree. Because ETT aggregates typically track vertex-specific properties, the subtrees rooted at any edge element will always evaluate to the identity element. Therefore, we do not need to physically store or maintain aggregate fields in the edge elements at all. We can only allocate memory for the aggregate fields in the $n$ vertex elements, dramatically reducing the overall memory footprint while preserving the update cost.

\myparagraph{Experimental Results}
We study this experimentally evaluate this optimization by constructing an ETT with $n=10^8$ and augmented values of size 32 bytes.
The space usage of MOJOS-ETT with this optimization was approximately 18.1GB. The regular MOJOS-ETT without the optimization used approximately 30.9GB.
For comparison, the SL-ETT (skip list) used approximately 38.1GB.
As augmented value size increases, the percentage space savings of this optimization also increases, because the space in augmented values begins to dominate the space in the ETT structure itself.
\section{Batch-Parallel Link-Cut Trees}

\subsection{Heavy Element Set} \label{app:lct_child_set}

Here we describe a data structure that maintains a dynamic set $S$ of non-negative integers in the range $[1,n)$ and supports insertions, deletions, and \textit{heavy element} queries (finding $x \in S$ such that $x \ge \frac{1}{2} \sum S$). The update operations run in expected $O(1)$ time, while queries run in worst-case $O(1)$ time. The space usage is $O(|S|)$.

The structure partitions elements into buckets of geometrically increasing size class $B_0, B_1, \dots, B_{\lfloor \log n \rfloor}$. A bucket $B_k$ stores all elements $x \in S$ satisfying $2^k \le x < 2^{k+1}$. We maintain:
\begin{enumerate}
    \item An array of hash tables representing the non-empty buckets $B_k$.
    \item A bitmap $M$ where the $k$-th bit is 1 iff $B_k$ is non-empty. Since there are $\lfloor \log n \rfloor$ size classes, this fits into a single word.
    \item The total sum $T = \sum_{x \in S} x$.
\end{enumerate}

To insert or delete an element $x$, we compute its bucket index $k = \lfloor \log_2 x \rfloor$. Using the bitmap we can compute the index of the bucket in the array in $O(1)$ time. We update the hash table $B_k$, the bitmap $M$, and the total sum $T$. Since we use hash tables, the cost per insertion and deletion is $O(1)$ in expectation.

To find a heavy element, first find the index $k$ of the largest non-empty bucket using the most significant bit of $M$. If $|B_k| \ge 4$, return $\vname{null}$. If $|B_k| < 4$, iterate through the elements in $B_k$. If any element $x$ satisfies $2x \ge T$, return $x$. Otherwise, return $\vname{null}$.
Since the algorithm iterates through $B_k$ only when $|B_k| < 4$, the query time is worst-case $O(1)$.

The total space complexity is proportional to the number of elements stored, $O(|S|)$. The space of the array of hash tables is $O(|S|)$ since it only stores non-empty buckets. The primary storage consists of the hash table nodes, and since each element $x \in S$ resides in exactly one bucket's hash table, the total space used by all tables is linear in $|S|$.

\myparagraph{Correctness}
The query algorithm's correctness relies on two observations regarding a heavy element $x$:
\begin{lemma}
If a heavy element $x$ exists, it must reside in the largest non-empty bucket $B_k$.
\end{lemma}
\begin{proof}
Suppose $x \in B_k$ is heavy, but there exists a non-empty bucket $B_j$ with $j > k$. Any element $y \in B_j$ satisfies $y \ge 2^{k+1} > x$. Since $T \ge x+y > 2x$, the condition $2x \ge T$ is violated. Thus, no such $y$ can exist.
\end{proof}

\begin{lemma}
If the largest non-empty bucket $B_k$ contains 4 or more elements, no heavy element exists.
\end{lemma}
\begin{proof}
Let $|B_k| = m$. Since $\forall y \in B_k, y \ge 2^k$, the total sum $T \ge m \cdot 2^k$.
For $x \in B_k$ to be heavy, we require $2x \ge T$. Since $x < 2^{k+1}$, we have:
\[ 2(2^{k+1}) > 2x \ge T \ge m \cdot 2^k \implies 2^{k+2} > m \cdot 2^k \implies 4 > m \]
Thus, if $m \ge 4$, a heavy element is mathematically impossible.
\end{proof}

\myparagraph{Parallel Heavy Element Set}
We extend the heavy element set data structure to support batches of $k$ insertions or deletions, in $O(k)$ work in expectation and $O(\log k)$ depth with high probability. The space usage is still $O(|S|)$. The core component is parallel hash tables to allow batch insertion and batch deletion. We maintain:

\begin{enumerate}
    \item An array of parallel hash tables~\cite{gil1991towards} representing the buckets $B_0, \dots, B_{\lfloor \log n \rfloor}$.
    \item A set of atomic integers representing the size $|B_i|$ of each bucket and the total sum $T$.
    \item A global bitmap $M$ indicating non-empty buckets.
\end{enumerate}

To process a batch of $k$ insertions or deletions, we perform a parallel reduction over the input batch to calculate the aggregate change to the total sum and the net change in element count for each specific bucket size class. Once computed, we add these aggregates to the global total sum $T$ and the respective bucket counters $|B_i|$.
Then we group the insertions and deletions by bucket and batch insert or batch delete each group in their respective parallel hash tables. Finally, we refresh the bitmap $M$, which can be done with a parallel reduction and a single write.

The work complexity is dominated by the hash table operations and reductions. The parallel reductions over the input batch requires $O(k)$ work, and the hash table updates take expected $O(k)$ work~\cite{gil1991towards}. The depth is dominated by the hash table batch insertions and batch deletions which is $O(\log k)$ with high probability in the binary-forking model.
The space complexity and correctness follow similarly from the sequential data structure.

\subsection{Supporting Reverse for Treaps} \label{app:lct_reverse}

Link cut-trees require this function for the dynamic sequence in addition to the standard MOJOS interface:
\begin{itemize}[leftmargin=15pt]
    \item $\reverse(x)$: given $x$, the representative element of a sequence, reverse the order of that sequence's elements.
\end{itemize}
Explicitly reversing the entire auxiliary tree representing a preferred path is typically too costly.
Instead, the splay tree implementation uses a lazy strategy where each internal node in the auxiliary tree contains a \defn{reversal bit} which indicates whether the subtree rooted at that node is logically reversed.
If the reversal bit of a node is set, it can be propagated down to the node's children by swapping the left and right child pointers, inverting the reversal bits of both children, and unsetting the reversal bit of this node.
Any time a node is accessed during an auxiliary tree operation, the reversal state of all of its ancestors must have been propagated down. This can be done top-down by starting at the root node of the auxiliary tree, and traversing down to the node of interest, propagating the reversal bit to the children at each node on the path.
At each node the bit is propagated downward by swapping the left and right child pointers, unsetting the reversal bit of the current node, and inverting the reversal bits of both children (if they exist).
We now describe how to extend this technique to treaps, batch-parallel treaps, and other binary tree data structures.

Sequentially, the same reversal bit technique can be used for treaps and any other binary tree data structure. Operations on elements in the binary tree must first travel to the root, and then propagate the set reversal bits back down to the initial element.
Thus, the cost of any operations is at least asymptotically as high as the depth of the element (or elements) in the tree. For treaps, zipzip trees, and biased zipzip trees this does not change the asymptotic cost of most dynamic sequence operation. In fact the exact work only increases by a constant factor.
The exception to this is $\getpred$ and $\getsucc$. If the treap stores direct pointer to its predecessor and successor, the normal algorithm can just follow that pointer in $O(1)$. With reversal bits, the downward propagation must also swap the predecessor and successor pointers. The cost of the operation increases to be asymptotically proportional to the depth of the element, since the reversal bits must be propagated down from the root to this element.

To support the $\reverse$ operation in batch-parallel treaps, zipzip trees, and biased zipzip trees, we introduce this helper function:
\begin{itemize}[leftmargin=15pt]
    \item $\fname{BatchPropagateReversal}(\vname{elements})$: for each element in the input batch, propagate down all of the reversal bits on the paths from the elements to their roots. This ensures that every node that is an ancestor of one of the elements has a reversal bit of false, including the elements themselves.
\end{itemize}
This must be called prior to any other batch operation on the set of elements that are part of the input to the batch operation. For example, prior to $\batchjoin$, each element that appears in at least one join must be passed to $\fname{BatchPropagateReversal}$.
By a similar analysis of the recomputation of augmented values in treaps, the cost of $\fname{BatchPropagateReversal}$ does not impact the overall asymptotic cost of batch link and batch cut.

\subsection{Parallel LCT Queries} \label{app:lct_query}

Our read-only connectivity query algorithm works as follows. For vertices $u$ and $v$, find the root of their link-cut tree component by calling $\getrep$ in the current auxiliary tree, then $\gethead$, and then following the non-preferred parent pointer. Finally check the equality between the root of the components containing $u$ and $v$.

Our read-only path query algorithm works as follows. First find the lowest preferred path that overlaps both the path from $u$ to the root and the path from $v$ to the root, call this the LCA preferred path. One way this can be done is by writing a list of the heads of each preferred path in order that is traversed for both vertices, then iterating in lockstep from the back of both arrays until they no longer match.
There are three cases.
First, if both vertices directly appear in the LCA preferred path, the query can be answered using a range sum query on the auxiliary tree data structure for this preferred path.
Second, assume both vertices do not appear in the LCA preferred path. Then, for both $u$ and $v$, we traverse up until this LCA path, every time we enter a new preferred path, we add to the total path sum, the range sum from the first node we enter into the preferred path until the head of the path. Then we add the range sum between the two nodes in the LCA preferred path that the paths from $u$ and $v$ first reached in their upward traversal.
Third, if exactly one vertex appears in the LCA preferred path, we do the method before for the vertex that does not appear, summing the prefix range sum of every preferred path. Then we add the range sum in the LCA preferred path from the first node that was reached to the node that is directly in the LCA preferred path.

Batch connectivity queries and batch get link-cut tree root queries work as follows. Consider the nested auxiliary tree data structures as one large tree. We can use a similar algorithm as for skip lists~\cite{tseng2019batch}. First for each query traverse up to the root and CAS a mark field on each node to true. If you win the CAS, proceed upwards, else the thread terminates. Then for each root node, proceed top-down. For each marked child, recursively in parallel, remove the mark and continue to your children. Once you reach a leaf node that was queried, write the name of the original root to answer that query.

\myparagraph{Analysis}
For the sake of analysis, we define the \defn{root path} of a vertex in the link-cut tree as the path followed through the nested zipzip tree structure from a vertex to reach the root.
This path is found by starting, at the vertex, traversing to the root of the current zipzip tree, then going to the non-preferred parent of the head of this zipzip tree (we assume the head can be reached from the root in $O(1)$ time which is possible to implement with augmented values), and repeating until the root is reached.

We first prove Lemma~\ref{lem:lct_zipzip_depth}, which shows that the length of any root path in the link-cut tree implemented with zipzip trees is logarithmic in the size of its component (with high probability). This lemma directly implies an $O(\log n)$ with high probability cost of connectivity queries.

\begin{lemma}\label{lem:lct_zipzip_depth}
    Let $T$ be a link-cut tree with the heavy preferred child invariant implemented using biased zip-zip trees. The length of any root path is $O(\log n)$ with high probability, where $n$ is the number of vertices in the component.
\end{lemma}
\begin{proof}
    We assume that augmented values allow finding the head of a preferred path in constant time. Consider the sequence of preferred paths traversed to reach the root from $v$. Let the non-preferred parents connecting these paths be $np_0=v, np_1, \dots, np_k=r$. By the heavy preferred child invariant, $k \le \lceil \log_2 n \rceil$.
    The total length of the root path is bounded by $k$ plus the sum of the depths of each $np_i$ within its respective auxiliary biased zip-zip tree.
    
    The expected depth of element $np_i$ in its auxiliary zip-zip tree is $O(\log(W_i/w(np_i)))$, where $W_i$ is the total weight of the auxiliary tree. Because this auxiliary tree is a child of $np_{i+1}$ via a non-preferred edge, its total weight $W_i$ is at most the weight $w(np_{i+1})$. Therefore, the expected depth of $np_i$ is $O(\log(w(np_{i+1})/w(np_i)))$.
    
    By linearity of expectation, the expected total number of nodes visited across all auxiliary trees on the root path is the sum of these expected depths. This forms a telescoping sum:
    $$ \sum_{i=0}^{k-1} O\left(\log \frac{w(np_{i+1})}{w(np_i)}\right) = O\left(\log \frac{w(r)}{w(v)}\right) \le O(\log n) $$
    
    To establish the high probability bound, we consider the global sum of all the depths across the entire root path.
    Because the total length of the path is a sum of negatively associated Bernoulli trials with expected value $O(\log n)$, standard Chernoff bounds apply to the global sum. The probability that the total length exceeds $c \log n$ decays exponentially with the constant $c$. By choosing a sufficiently large constant, the total length of the root path is bounded by $O(\log n)$ with high probability in $n$.
\end{proof}

A similar bound may be proved for path queries, also accounting for the range sum queries needed.
We prove Lemma~\ref{lem:range_query_cost}. Combining this with similar reasoning as Lemma~\ref{lem:lct_zipzip_depth} proves that path queries take $O(\log n)$ time with high probability.

\begin{lemma} \label{lem:range_query_cost}
    Range sum queries between elements $x$ and $y$ in a biased zipzip tree with total weight $W$ can be done in $O(\log(W/w(x)) + \log(W/w(y)))$ time with high probability. A suffix or prefix sum query from element $x$ takes $O(\log(W/w(x)))$ with high probability.
\end{lemma}
\begin{proof}
    Augment the biased zipzip tree so that each node stores the sum of the weights in its subtree. A suffix sum on element $x$ can be computed by starting a cumulative sum with the weight of $x$, then traversing up to the root, adding to the cumulative sum the subtree weight of every right child of nodes on this path to the root. Since the depth of $x$ is $O(\log(W/w(x)))$ with high probability, that is also the cost of this operation. Prefix sums can be done symmetrically with the same cost bound.

    Any range sum between $x$ and $y$ can be computed as follows. First determine if $x$ is to the left of $y$ or vice versa (by extracting their paths to the root, finding the highest node that differs, and checking which one is the left/right child of the node just above that).
    Without loss of generality assume that $x$ is to the left of $y$.
    Now compute the suffix sum of $x$, and the suffix sum of $y$. Subtract the latter from the former, and add the weight of $y$ back. This gives the range sum from $x$ to $y$. It is clear that the total number of nodes traversed is proportional to the sum of the depths of $x$ and $y$ which is $O(\log(W/w(x)) + \log(W/w(y)))$ with high probability.
\end{proof}

We also describe a batch suffix sum query algorithm.
This also can extend to batch prefix sums and batch range sums.
Although this does not lead to a better bound for batch path queries, we use it later in our analysis.
Similar to batch connectivity queries, the algorithm has a bottom-up marking phase and a top-down phase.

In the bottom-up marking phase, each of the $k$ query nodes initiates a traversal upward toward the root. At each visited node, the process attempts to change a boolean mark flag from false to true using an atomic compare-and-swap operation. If the compare-and-swap succeeds, the traversal continues upward to the parent node. If it fails, it indicates that another concurrent query path has already visited and marked this node, meaning all ancestors up to the root are also guaranteed to be marked. The current traversal can safely terminate early.

In the top-down accumulation phase, a single parallel traversal begins at each root with an initial cumulative suffix weight of zero. The traversal recursively forks and descends only into children that have their mark flag set to true. When descending to a right child, the cumulative weight parameter is passed down entirely unchanged. When descending to a left child, the traversal adds the individual weight of the current node and the total subtree weight of its right child to the cumulative parameter, because all of those elements lie strictly to the right of any keys in the left subtree.

When the recursion reaches one of the original $k$ query nodes, the accumulated parameter holds the exact total weight strictly to the right of that key. During this top-down return path, the algorithm can also safely revert the mark flags to false, restoring the tree to a clean state for subsequent batch operations without requiring any additional passes.

\begin{lemma}\label{lem:batch_range_query}
    Given $k$ query nodes in a biased zip-zip tree, the two-phase batch suffix-weight algorithm correctly computes the suffix weights for all $k$ nodes. The total expected work is $O(|U|)$, where $U$ is the union of the search paths, and the parallel depth is $O(D)$, where $D$ is the maximum depth among the $k$ query nodes.
\end{lemma}
\begin{proof}
    We analyze the total work and depth of the two-phase batch suffix-weight algorithm on a biased zip-zip tree. Let $U$ denote the geometric union of the paths from the $k$ query nodes to the root, and let $D$ be the maximum depth of any node among the $k$ query nodes. 
    
    During the bottom-up marking phase, $k$ independent processes traverse upward. A process performs a constant amount of work at each node to execute an atomic compare-and-swap operation on the mark flag. If the operation succeeds, the process advances to its parent. If it fails, the node was already marked by a concurrent process, guaranteeing that all ancestors up to the root are also marked, which allows the current process to safely terminate. Because each node in $U$ transitions from unmarked to marked exactly once, there are exactly $|U|$ successful atomic operations. Furthermore, each of the $k$ paths can experience at most one failed atomic operation before terminating. Therefore, the total work for the bottom-up phase is strictly bounded by $O(|U| + k)$, which simplifies to $O(|U|)$ since the union of paths must contain at least the $k$ query nodes. 
    
    The top-down accumulation phase begins at the root and only visits children that have their mark flag set to true. By definition, the set of marked nodes is exactly $U$. The algorithm performs a constant amount of work at each visited node to compute the running suffix sum, propagate it to the appropriate marked children, and revert the mark flag to false. Since the recursion is restricted entirely to the nodes in $U$, the work for the top-down phase is also $O(|U|)$. Thus, the total expected work is strictly proportional to the size of the union of the search paths. 
    
    For the parallel depth, the bottom-up phase processes nodes sequentially along each path. The maximum number of steps any single process can take is bounded by the maximum distance from a query node to the root, which is exactly $D$. Because the atomic operations take constant time, the depth of the marking phase is $O(D)$. 
    
    The top-down phase traverses the marked nodes recursively. The critical path of this recursion follows the longest marked path from the root down to a query node. Since the parallel forks at each node require constant time, the depth of the accumulation phase is also $O(D)$. Therefore, the total parallel depth of the algorithm is strictly proportional to the maximum depth of a query node.
\end{proof}

Now we prove Lemma~\ref{lem:lct_zipzip_path_overlap}, which bounds the total number of distinct nodes in the union of $k$ root paths.
In combination with Lemma~\ref{lem:lct_zipzip_depth}, this proves that the batch connectivity query algorithm (and batch root find algorithm) takes $O(k \log (1+n/k))$ work with high probability and $O(\log n)$ depth with high probability.


\begin{lemma} \label{lem:lct_zipzip_path_overlap}
    In a link-cut tree with the heavy preferred child invariant implemented with biased zipzip trees, the number of distinct nodes on the union of $k$ root paths is $O(k \log(1 + n/k))$ \revision{in expectation}.
\end{lemma}

\revision{

\begin{proof}
    Let $\mathcal{M}$ denote the nested tree formed by the non-preferred edges from each root node to its non-preferred parent, and the internal edges of the biased zipzip trees. From our previous results, the expected length of a path in $\mathcal{M}$ from any descendant $b$ to any ancestor $t$ is bounded by $O(\log(s(t) / s(b)))$, where $s(v)$ denotes the LCT subtree size of $v$.
    
    The union of $k$ root paths in the link-cut tree corresponds to a connected subtree $\mathcal{T}$ in the nested tree $\mathcal{M}$ containing the root and at most $k$ leaves. We decompose $\mathcal{T}$ into disjoint maximal path segments $Q_1, \dots, Q_m$ by splitting it at every branching node. Since a tree with $k$ leaves branches at most $k-1$ times, there are $m \le 2k-1$ such segments. Let each segment $Q_j$ start at a bottom node $b_j$ and end at a top node $t_j$. The total expected number of distinct nodes is:
    \[
        E[|\mathcal{T}|] \le \sum_{j=1}^m O\left(\log \frac{s(t_j)}{s(b_j)}\right)
    \]
    
    We bound this sum using bottom-up induction on the branching structure. Let $f(n, \ell)$ be the maximum sum of expected segment lengths for a subtree of $\mathcal{T}$ containing $\ell$ leaves, rooted at a node with LCT size $n$.
    
    The base case ($\ell=1$) is when the nested tree is a single path. The sum is simply bounded by $\log(n / s(b_1)) \le \log n$.
    
    Now we explain the inductive recursion. Suppose the root path of this subtree goes down to a branching node of size $n'$, then splits into $d$ subtrees. Let the $i$-th subtree contain $\ell_i$ leaves (where $\sum_{i=1}^d \ell_i = \ell$) and have a root of LCT size $n_i$. Because the subtrees are disjoint in the link-cut tree, we have $\sum_{i=1}^d n_i \le n' \le n$.
    
    We can write the recurrence as:
    \[
        f(n, \ell) \le \log\left(\frac{n}{\sum_{i=1}^d n_i}\right) + \sum_{i=1}^d f(n_i, \ell_i)
    \]
    
    Assume inductively that $f(n, \ell) \le \ell \log(n/\ell)$. Substituting this into the recurrence gives:
    \[
        f(n, \ell) \le \log\left(\frac{n}{\sum_{i=1}^d n_i}\right) + \sum_{i=1}^d \ell_i \log\left(\frac{n_i}{\ell_i}\right)
    \]
    
    By the concavity of the logarithm function (Jensen's Inequality), the summation is bounded by:
    \[
        \sum_{i=1}^d \ell_i \log\left(\frac{n_i}{\ell_i}\right) 
        \le \ell \log\left(\frac{\sum_{i=1}^d n_i}{\ell}\right)
    \]
    
    Let $x = \sum_{i=1}^d n_i$. Our expression simplifies to:
    \[
        f(n, \ell) \le \log\left(\frac{n}{x}\right) + \ell \log\left(\frac{x}{\ell}\right)
    \]
    
    This expression is maximized when $x$ is as large as possible, which is $x = n$. Plugging in $x = n$ yields:
    \[
        f(n, \ell) \le \log\left(\frac{n}{n}\right) + \ell \log\left(\frac{n}{\ell}\right) = \ell \log\left(\frac{n}{\ell}\right)
    \]
    
    Applying this induction to the root of the entire link-cut tree where the total number of leaves is $\ell = k$, the expected number of distinct nodes on the union of $k$ root paths is $O(k \log(n/k))$. To handle the edge case where $k$ approaches or exceeds $n$, it is standard practice to write the final bound as:
    \[
        O\left(k \log\left(1 + \frac{n}{k}\right)\right).
    \]
\end{proof}

}

\subsection{LCT Batch Cut Algorithm Details} \label{app:lct_details}

Algorithm~\ref{alg:lct_batch_cut} shows the pseudo-code for our batch-cut algorithm.

\begin{algorithm}[ht]
\caption{$\fname{BatchCut}(\vname{cuts})$}
\label{alg:lct_batch_cut}
    $NTE \gets$ non-preferred edges in $\vname{cuts}$. \;
    $TE \gets \vname{cuts} \setminus NTE$ \;
    \parfor{$(u,v) \in NTE$} {
        $u.nppar = \vname{null}$
    }
    $\batchcut(NTE)$ \;
    Update weights of vertices. \;
    \parfor{$[v, \vname{group}] \in \fname{GroupByDst}(NTE)$} {
        $v.\fname{DeleteChildren}(\vname{group})$
    }
    \textit{(Restore the heavy preferred child invariant.)}
\end{algorithm}

During $\batchcut$ there may be some vertices violating the invariant that were not them self reweighted.
Specifically, for a vertex $v$ with current preferred child $c$, if something in the LCT subtree of $c$ is cut, both $n(c)$ and $n(v)$ decrease by the same amount.
By the invariant, previously $2\cdot n(c) > n(v)$, but if they both decrease by an equal amount, this property can be violated.

Thus we need to efficiently find all the \defn{light edges} whose current preferred child is not the heavy child.
Sleator and Tarjan~\cite{sleator1983data} described how to find the light edges in a preferred path sequentially, using an operation they call \defn{conceal}.
They proved that the operation takes $O(\log n)$ time.
In our algorithm the cost is $O(\log n)$ \whp{} since we use biased zipzip trees.

In the batch-parallel setting, we can run the conceal operation in parallel for each preferred path that was affected by the batch cut.
Since each conceal takes $O(\log n)$ \whp{}, the depth of this is step is $O(\log n)$ \whp{}.
In Appendix~\ref{app:lct_cost} we prove that the total work of this step is $O(k \log (1+n/k))$ \whp{}.

The remaining steps over the simple batch cut algorithm are similar to the additions to the batch link algorithm. We determine the vertices that require weight updates, update child sets, then check which vertices need to change their preferred child, then change those preferred children with a $\batchcut$ followed by a $\batchlink$.

\subsection{LCT Batch Update Correctness Proofs} \label{app:lct_correct}

Here we prove Lemma~\ref{lem:only_reweighted_change_pref}. This helps to justify the correctness of our batch update algorithms. Since our algorithms examine all vertices that were reweighted, and update their preferred child using their heavy element child set, both algorithms correctly maintain the heavy preferred child invariant.
The correctness of maintaining connectivity information and answering queries is easy to verify.

\begin{lemma} \label{lem:only_reweighted_change_pref}
    In a link-cut tree maintaining the heavy child invariant, a vertex can only become a candidate to change its preferred child during $\batchlink$ if it undergoes a weight update.
\end{lemma}

\begin{proof}
    We proceed by contradiction. Assume there exists a vertex $v$ that was not reweighted during a batch of operations, yet it requires a change to its preferred child. Let $c$ be the preferred child of $v$ prior to the operations. Because the tree initially maintained the heavy child invariant, $c$ was a heavy child, meaning $2n(c) > n(v)$. Using the definition $n(v) = n(c) + w(v)$, this inequality simplifies to $n(c) > w(v)$, where $w(v)$ represents the total size of all non-preferred subtrees of $v$ plus one.
    
    For $v$ to require a change to its preferred child after the operations, the invariant must be violated. This can occur in exactly two ways. The first case is that $c$ is no longer a heavy child, meaning its updated size $n'(c)$ and the updated weight $w'(v)$ satisfy $n'(c) \le w'(v)$. In batch link, the size of any vertex can only increase, so $n'(c) \geq n(c)$.
    In addition, since $v$ was not reweighted, $w'(v) = w(v)$.
    Combining these yields $n(c) \leq n'(c) \leq w'(v) = w(v)$. This directly contradicts the established initial state where $n(c) > w(v)$.
    
    The second case is that some non-preferred child $d$ of $v$ becomes heavy, requiring $v$ to switch its preferred edge to $d$. However, since $v$ was not reweighted, this is not possible.
\end{proof}

\subsection{LCT Batch Update Cost Bound Proofs} \label{app:lct_cost}

We now prove Theorem~\ref{thm:lct_analysis} for the efficiency of batch link and batch cut (separated into two parts for readability).
We proved the efficiency of queries in Appendix~\ref{app:lct_query}.

\lctbatchupdate*
\begin{proof}
    We analyze the work and depth of $\batchlink$. 
    The $\fname{RootLinkTrees}$ operation starts with a $\batchgetrep$ call which takes $O(k \log (1+n/k))$ \revision{expected} work and $O(\log n)$ depth with high probability by Lemma~\ref{lem:lct_zipzip_path_overlap}.
    The Euler tour technique over the link tree with $k$ edges takes $O(k)$ work and $O(\log k)$ depth~\cite{tarjan1985efficient}.
    Each of the $k$ $\evert$ operations happens independently in parallel and each one takes $O(\log n_i)$ time where $n_i$ is the size of the link-cut tree component it occurs in. Since the link-cut tree components are disjoint, $\sum_i n_i \leq n$. Then applying Lemma 2.1 from~\cite{acar2019parallel}, the sum of these costs is $O(k \log (1+n/k))$ \revision{in expectation}. The depth is bottle-necked by a single call to $\evert$, which is $O(\log n)$ with high probability.
    Grouping the non-preferred edges by destination via $\fname{GroupByDst}$ and calling $\fname{InsertChildren}$ takes $O(k)$ expected work and $O(\log n)$ depth \whp{} using parallel semisort~\cite{gu2015top}, and the batch insert of the parallel heavy set.
    
    Now we bound the costs of the weight updating step.
    Note that both vertices that are directly weight updated and indirectly weight updated, lie on the root paths of endpoints of the $k$ links.
    Lemma~\ref{lem:lct_zipzip_path_overlap} proved that the number of nodes on the union of $k$ root paths is $O(k \log (1+n/k))$ \revision{in expectation}, so this is also an upper bound for the number of weight updated vertices.
    Therefore our Euler tour technique to compute the subtree sums takes $O(k \log (1+n/k))$ \revision{expected} work and $O(\log(k \log (1+n/k)) = O(\log n)$ depth with high probability.
    To bound the work for updating the augmented values of weights, note that for every reweighted vertex, its path to the root of its biased zipzip tree also falls on the root path of one of the endpoints of the $k$ links. Therefore the total cost in recomputing all of these augmented values is also $O(k \log (1+n/k))$ \revision{expected} work, and $O(\log n)$ depth with high probability.

    The next part iterates over all the reweighted vertices in parallel, calling $\fname{GetHeavyChild}$ and $\getsucc$. We will assume that $\getsucc$ can run in $O(1)$ time which is easy to maintain in the biased zipzip tree by maintaining a direct pointer to the successor of each element.
    We use the batch-parallel suffix sum query to compute all of the $\fname{GetHeavyChild}$ calls. The cost of this is proportional to the number of nodes on the union of the paths from each of these vertices to their zipzip tree root by Lemma~\ref{lem:batch_range_query}.
    Since these are a subset of the root paths of the endpoints of the $k$ links, this step also takes $O(k \log (1+n/k))$ \revision{expected} work, and $O(\log n)$ depth with high probability.

    Finally, the algorithm calls $\batchsplit$ and $\batchjoin$ to update preferred children. The cost of these operations is proportional to the number of nodes on the union of all the spines in the biased zipzip trees that they impact. We have already observed that the paths from each reweighted vertex to their root in their biased zipzip tree are a subset of the root paths of the endpoints of the $k$ links. This is $O(k \log (1+n/k))$ nodes \revision{in expectation}. Thus the work is $O(k \log (1+n/k))$ work \revision{in expectation}, and the depth is $O(\log n)$ with high probability.

    We analyze the work and depth of $\batchcut$. 
    The algorithm begins by partitioning the $k$ input edges into non-preferred edges ($NTE$) and preferred edges ($TE$). Updating the non-preferred parent pointers $u.nppar$ for all $(u,v) \in NTE$ takes $O(k)$ work and $O(\log k)$ depth. Performing the base structural cuts on the preferred edges takes $O(k \log(1+n/k))$ work \revision{in expectation} and $O(\log n)$ depth with high probability. Grouping the non-preferred edges by destination via $\fname{GroupByDst}$ and calling $\fname{DeleteChildren}$ takes $O(k)$ expected work and $O(\log k)$ depth using parallel semisort~\cite{gu2015top}, and the batch delete of the parallel heavy set.

    Now we bound the costs of the weight updating step. 
    Similar to the $\batchlink$ algorithm, both the directly and indirectly weight-updated vertices lie on the root paths of the endpoints of the $k$ cut edges. 
    Lemma~\ref{lem:lct_zipzip_path_overlap} proved that the number of nodes on the union of $k$ root paths is bounded by $O(k \log (1+n/k))$ \revision{in expectation}, which bounds the number of weight-updated vertices. 
    Therefore, updating the augmented values of weights for these vertices takes $O(k \log (1+n/k))$ work \revision{in expectation}, and $O(\log n)$ depth with high probability, because their paths to the root of their biased zip-zip trees fall completely on the root paths of the cut endpoints.

    An additional step is to find the light vertices that were not reweighted. As described in Appendix~\ref{app:lct_details}, each conceal operation takes $O(\log n)$ depth \whp{} and runs in parallel, so the total depth of this step is $O(\log n)$ \whp{}.
    Now we prove that the total work of these conceal operations is $O(k \log (1+n/k))$ \revision{in expectation}.
    First recall that light vertices can only be ancestors of one of the cut edges. Thus the conceal operations essentially operate along the $k$ root paths based at the upper vertex of each cut.
    Since the conceal operations do not overlap, their total cost is essentially the number of nodes on the union of $k$ root paths. We have previously shown that this is $O(k \log (1+n/k))$ \revision{in expectation}.

    The next part iterates over all the reweighted vertices or light vertices in parallel, calling $\fname{GetHeavyChild}$ and $\getsucc$. We will assume that $\getsucc$ can run in $O(1)$ time, which is easy to maintain in the biased zip-zip tree by keeping a direct pointer to the successor of each element. 
    We use the batch-parallel suffix sum query to compute all of the $\fname{GetHeavyChild}$ calls. The cost of this is proportional to the number of nodes on the union of the paths from each of these vertices to their zip-zip tree root by Lemma~\ref{lem:batch_range_query}. 
    Since these are a subset of the root paths of the endpoints of the $k$ cuts, this step also takes $O(k \log (1+n/k))$ work \revision{in expectation}, and $O(\log n)$ depth with high probability.

    Finally, the algorithm calls $\batchsplit$ and $\batchjoin$ to update preferred children for vertices where the heavy child has changed. The cost of these operations is proportional to the number of nodes on the union of all the spines in the biased zip-zip trees that they impact. We have already observed that the paths from each reweighted vertex to their root in their biased zip-zip tree are a subset of the root paths of the endpoints of the $k$ cuts. This is $O(k \log (1+n/k))$ nodes \revision{in expectation}. Thus the work is $O(k \log (1+n/k))$ \revision{in expectation}, and the depth is $O(\log n)$ with high probability.
\end{proof}

\myparagraph{Diameter Bounds}
Prior work~\cite{deman2026ufo} has proven a loose $O(D^2)$ worst-case cost bound for link-cut trees. This comes from the fact that any root path can traverse through at most $D$ preferred paths, and each preferred path has at most $D$ nodes. Using splay trees, it is possible that the path from a node to the root of its auxiliary tree traverses through all of the nodes in the splay tree, so the worst-case cost is bounded by $O(D^2)$.
This bound can be improved to an expected $O(D \log D)$ when using treaps.

\begin{lemma} \label{lem:expected_root_path}
    In a link-cut tree implemented with treaps, the expected length of any root path is $O(D \log D)$, where $D$ is the diameter of the input tree.
\end{lemma}
\begin{proof}
    In the link-cut tree, any path from a node to the root traverses at most $D$ preferred paths. Each of these preferred paths contains at most $D$ nodes.
    For a treap of size $s$, the expected depth of any node is $O(\log s)$. Since the size of any preferred path is bounded by $s \le D$, the expected traversal length within a single auxiliary tree is $O(\log D)$.
    Because a root path can transition between at most $D$ preferred paths, the traversal in the link-cut tree passes through at most $D$ auxiliary trees. By linearity of expectation, the total expected length of the root path in the link-cut tree is bounded by the sum of expected depths across these auxiliary trees, which evaluates to $D \cdot O(\log D) = O(D \log D)$.
\end{proof}

\revision{In fact, just using naive linked lists can result in root paths with length bounded by $O(D)$. We will use this simple data structure for our analysis.}
Now we prove Theorem~\ref{thm:simple_lct_analysis} which proves bounds on the simple batch link-cut tree algorithms (with no heavy preferred child invariant) in terms of the diameter of the input tree.
We also prove in Theorem~\ref{thm:lct_diameter_analysis} that similar bounds hold for the algorithms which do maintain the heavy preferred child invariant, and use biased zipzip trees.

\lctbatchupdatesimple*
\begin{proof}
    In the simple batch link-cut tree, we analyze the cost of the distinct phases of the algorithms.
    \revision{
    Our analysis will assume that we are using naive linked list based implementation of batch-dynamic sequences.
    }
    
    For $\batchlink$, the algorithm first computes a rooted forest over the components using the Euler tour technique.
    The components are found with a batch get root query on the link-cut tree which takes at most work proportional to length of the $k$ root paths.
    \revision{This is $O(k \cdot D)$ in total if we are using linked lists. The depth is similarly $O(D)$ for this step.}
    The Euler tour technique itself takes $O(k)$ work and $O(\log k)$ depth. Next, the algorithm performs at most one $\evert$ operation per connected component. Because these components are strictly disjoint, the $\evert$ operations can execute completely in parallel without contention.
    \revision{Using linked lists, each $\evert$ operation traverses $O(D)$ nodes. Across the $k$ components, this requires $O(k \cdot D)$ work.
    The depth is dominated by the longest parallel $\evert$ traversal, plus the cost of spawning $k$ threads bounding the overall $\batchlink$ depth as $O(D + \log k)$.}

    For $\batchcut$, the algorithm first identifies and clears non-preferred edges by checking parent pointers, which takes $O(k)$ work and $O(\log k)$ parallel depth.
    The remaining preferred edges are removed by executing a $\batchsplit$ on the auxiliary trees. Because each auxiliary tree represents a path of length at most $D$, \revision{these splits require at most $O(k \cdot D)$ work and $O(\log k + D)$ depth using linked lists.}
    Finally, query operations traverse a root path, taking \revision{$O(D)$ time}.
\end{proof}

\begin{theorem} \label{thm:lct_diameter_analysis}
    In a link-cut tree with the heavy preferred child invariant implemented with biased zipzip trees, the batch link and batch cut algorithms take $O(k D^2)$ expected work and $O(D^2 + \log k)$ \revision{expected} depth where $k$ is the batch size and $D$ is the diameter of the input tree. Queries take $O(D^2)$ time.
\end{theorem}
\begin{proof}
    We first analyze the work and depth of $\batchlink$. 
    The $\fname{RootLinkTrees}$ part first calls $\batchgetrep$. In biased zipzip trees used for auxiliary trees, a node has a worst-case depth of $O(D)$. Since a root path can transition between at most $D$ preferred paths, traversing up to $D$ such auxiliary trees yields a worst-case root path length bounded by $O(D^2)$. Thus, for $k$ endpoints, $\batchgetrep$ takes $O(k D^2)$ work and $O(D^2 + \log k)$ depth.
    The Euler tour technique over the link tree with $k$ edges takes $O(k)$ work and $O(\log k)$ depth~\cite{tarjan1985efficient}.
    Each of the $k$ $\evert$ operations happens independently in parallel. Since each $\evert$ traverses a root path, it takes $O(D^2 + \log k)$ depth. Summed across the disjoint components, the work is $O(k D^2)$. 
    Grouping the non-preferred edges by destination via $\fname{GroupByDst}$ and calling $\fname{InsertChildren}$ takes expected $O(k)$ work and $O(\log k)$ expected depth using parallel semisort~\cite{gu2015top}, and the batch insert of the parallel heavy set.
    
    Now we bound the costs of the weight updating step.
    Note that both vertices that are directly weight updated and indirectly weight updated lie on the root paths of the endpoints of the $k$ links. Because the worst-case maximum length of a single root path is $O(D^2)$, the number of nodes on the union of the $k$ root paths is strictly bounded by $O(k D^2)$, which tightly bounds the number of weight updated vertices.
    Therefore our Euler tour technique to compute the subtree sums takes $O(k D^2)$ work and $O(\log D + \log k)$ depth.
    To bound the work for updating the augmented values of weights, note that for every reweighted vertex, its path to the root of its biased zipzip tree also falls on the root path of one of the endpoints of the $k$ links. Therefore the total cost in recomputing all of these augmented values is also $O(k D^2)$ work, and $O(D^2 + \log k)$ depth.

    The next part iterates over all the reweighted vertices in parallel, calling $\fname{GetHeavyChild}$ and $\getsucc$. We will assume that $\getsucc$ can run in $O(1)$ time which is easy to maintain in the biased zipzip tree by maintaining a direct pointer to the successor of each element.
    We use the batch-parallel suffix sum query to compute all of the $\fname{GetHeavyChild}$ calls. The cost of this is proportional to the number of nodes on the union of the paths from each of these vertices to their zipzip tree root.
    Since these are a subset of the root paths of the endpoints of the $k$ links, this step also takes $O(k D^2)$ work, and $O(\log D + \log k)$ depth.

    Finally, the algorithm calls $\batchsplit$ and $\batchjoin$ to update preferred children. The cost of these operations is proportional to the number of nodes on the union of all the spines in the biased zipzip trees that they impact. We have already observed that the paths from each reweighted vertex to their root in their biased zipzip tree are a subset of the root paths of the endpoints of the $k$ links. This is $O(k D^2)$ nodes. Thus the work is $O(k D^2)$, and the depth is $O(D^2 + \log k)$. The overall depth for $\batchlink$ evaluates to \revision{expected} $O(D^2 + \log k)$, and the overall work is expected $O(k D^2)$.

    We now analyze the work and depth of $\batchcut$. 
    The algorithm begins by partitioning the $k$ input edges into non-preferred edges ($NTE$) and preferred edges ($TE$). Updating the non-preferred parent pointers $u.nppar$ for all $(u,v) \in NTE$ takes $O(k)$ work and $O(\log k)$ depth. Performing the base structural cuts on the preferred edges takes $O(k D^2)$ work and $O(D^2 + \log k)$ depth. Grouping the non-preferred edges by destination via $\fname{GroupByDst}$ and calling $\fname{DeleteChildren}$ takes expected $O(k)$ work and \revision{expected} $O(\log k)$ depth using parallel semisort~\cite{gu2015top}, and the batch delete of the parallel heavy set.

    Now we bound the costs of the weight updating step. 
    Similar to the $\batchlink$ algorithm, both the directly and indirectly weight-updated vertices lie on the root paths of the endpoints of the $k$ cut edges. 
    The number of nodes on the union of $k$ root paths is bounded by $O(k D^2)$, which bounds the number of weight-updated vertices. 
    Therefore, updating the augmented values of weights for these vertices takes $O(k D^2)$ work, and $O(D^2 + \log k)$ depth, because their paths to the root of their biased zip-zip trees fall completely on the root paths of the cut endpoints.

    The next part iterates over all the reweighted vertices $W$ in parallel, calling $\fname{GetHeavyChild}$ and $\getsucc$. 
    We use the batch-parallel suffix sum query for the $\fname{GetHeavyChild}$ calls. The cost of this is proportional to the number of nodes on the union of the paths from each of these vertices to their zip-zip tree root. 
    Since these are a subset of the root paths of the endpoints of the $k$ cuts, this step also takes $O(k D^2)$ work, and $O(D^2 + \log k)$ depth.

    Finally, the algorithm calls $\batchsplit$ and $\batchjoin$ to update preferred children for vertices where the heavy child has changed. The cost of these operations is proportional to the number of nodes on the union of all the spines in the biased zip-zip trees that they impact. Since these paths are a subset of the root paths of the endpoints of the $k$ cuts, it impacts $O(k D^2)$ nodes. Thus the work is $O(k D^2)$, and the depth is $O(D^2 + \log k)$. The overall depth for $\batchcut$ evaluates to \revision{expected} $O(D^2 + \log k)$, and the overall expected work is $O(k D^2)$. Queries traverse root paths of length at most $D$, taking $O(D^2)$ time.
\end{proof}

\section{Experimental Evaluation}\label{app:experiments}

\subsection{MOJOS Implementation Details}

\myparagraph{MOJOS Batch-Parallel Treap}
Each treap node stores only a parent pointer, two child pointers, and a 64 bit priority integer (32 total bytes per treap node).
Our treap implementation stores exactly $n$ nodes for a dynamic sequence with $n$ elements, so the memory usage is $32n$ bytes plus a small constant.

For our two-phase batch-split algorithm, the second phase must have access to each attempted child pointer change from the first phase. To store these, we use a single thread-local array per thread to store the child pointer changes that occurred on that thread. In the second phase, each thread sequentially iterates through the entries in its thread-local array.

\myparagraph{MOJOS ETT}
Our implementation uses the implementation of skip-list based batch-parallel Euler tour trees from~\cite{tseng2019batch} as a starting point.
We modify the code with the extra steps needed to function with batch-parallel treaps.
We use the cyclic batch-parallel treap algorithm (Appendix~\ref{app:cyclic_treap}) to handle cycle-inducing join batches.
For batch cut, we don't implement the theoretically efficient algorithm. Instead we use the practical algorithm from~\cite{tseng2019batch} which picks a large random subset of the remaining cuts to execute in each round, then recursing on the remaining cuts until all have been completed.

\myparagraph{MOJOS Simple LCT}
We implement our batch-parallel link-cut trees with batch-parallel treaps instead of batch-parallel biased zipzip trees which provide stronger theoretical guarantees.
We made this decision due to the fact that biased zipzip trees are cumbersome in practice, requiring extra space and update costs to maintain the weights of elements. Each time weights change, a call to $\batchupdate$ is required which is a costly operation.

Each node contain a treap parent pointer, a non-preferred parent pointer, two child pointers, a 32 bit priority integer, a boolean reversal bit value, and a boolean mark field used for batch-queries and efficient augmented value updates. Each node is padded to 40 bytes, thus the total memory usage of the simple batch-parallel link-cut tree is $40n$ bytes plus a small constant factor.

The key component of the simple batch link algorithm is the Euler tour technique step.
Our Euler tour technique implementation uses a technique similar to that in the implementation of~\cite{dong2023provably}.
First, each undirected edge is mapped into a pair of directed edges, which are cyclically linked at each vertex to construct closed Euler tours. This structure is flattened into an Euler tour array through a parallel list-ranking algorithm that utilizes $\sqrt{n}$ node sampling and prefix sums to compute global offsets.
We extend this algorithm to deal with the fact that our use case requires the Euler tour technique to work on a forest of disconnected trees.
Our implementation handles this by using a concurrent union-find~\cite{blelloch2014concurrent} to compute the connected components over the input, and find one vertex in each tree. This is necessary to ensure that at least one vertex is sampled in each tree.
In theory list contraction is used to efficiently compute this step, but we found concurrent union-find to be faster in practice.

\myparagraph{MOJOS Robust LCT}
This implementation stores two additional 32 bit integers per treap node: one for the weight of each vertex, and one to store the sum of the weights in its treap subtree.
Thus each node takes 48 bytes, and the total memory usage is $48n$ plus a small constant factor.

We use the same Euler tour technique implementation as the simple LCT to direct the links.
Then we use thread-local buffers (similar to batch split in our treap implementation) to collect all of the direct and indirect weight changes.
We use a second Euler tour technique to compute the subtree sums. This call does not require the concurrent union find to compute connected components, because the weight change propagation tree is naturally rooted.

\subsection{Real-World Graph Inputs}

Table~\ref{tab:input_graphs} summarizes the real-world graph inputs used for our experimental evaluation.

\begin{table}[ht]
\small
   \centering
   \begin{tabular}{l l l c c c}
       \toprule
       Name & Abbrv. & Type & $|V|$ & $|E|$ & Cite \\
       \midrule
       USA Roads & USA & Road & 23.95M & 28.85M & \cite{roadgraphs} \\
       ENWiki & ENW & Web & 4.21M & 91.94M & \cite{boldi2004webgraph} \\
       StackOverflow & SO & Temporal & 6.02M & 28.18M & \cite{paranjape2017motifs} \\
       Twitter & TWIT & Social & 41.65M & 1.20B & \cite{kwak2010what} \\
       \bottomrule
   \end{tabular}
   \caption{The real-world graph datasets in our experiments.}
   \vspace{-3em}
   \label{tab:input_graphs}
\end{table}

\end{document}